\documentclass[a4paper,11pt]{article}
\usepackage[utf8x]{inputenc}

\usepackage{a4wide}
\usepackage{amssymb}
\usepackage{amsmath}
\usepackage{amsthm}
\usepackage[mathscr]{euscript}
\usepackage[colorlinks=true]{hyperref}
\usepackage[usenames,dvipsnames]{xcolor}

\numberwithin{equation}{section}

\theoremstyle{definition}
\newtheorem{definition}{Definition}[section]

\theoremstyle{plain}
\newtheorem{Theorem}[definition]{Theorem}
\newtheorem{Proposition}[definition]{Proposition}
\newtheorem{Lemma}[definition]{Lemma}
\newtheorem{Corollary}[definition]{Corollary}

\theoremstyle{remark}
\newtheorem{remark}[definition]{Remark}

\newcommand{\R}{\mathbb R}
\newcommand{\N}{\mathbb N}

\newcommand{\eps}{\varepsilon}

\usepackage[xcolor]{changebar}
\cbcolor{blue}
\newcommand{\Ric}{\mathrm{Ric}}
\newcommand{\gec}{{\check g_\eps}}

\newcommand{\comp}{\Subset}

\newcommand{\sse}{\subseteq}

\newcommand{\linfloc}{L^\infty_{\mathrm{loc}}}
\newcommand{\D}{\mathcal{D}}
\newcommand{\diag}{\mathrm{diag}}
\newcommand{\lamin}{\lambda_{\text{min}}}

\usepackage[inline]{enumitem}
\newcommand{\enumlabelformat}{\roman}
\newcommand{\enumlabelfont}[1]{#1}
\newlength{\thelabelsep}
\setlist{labelsep=\thelabelsep}
\setlist[enumerate]{font=\enumlabelfont,label=(\enumlabelformat*),leftmargin=2.5em}
\setlist[itemize]{leftmargin=2.5em,label=$-$}

\newcounter{inlineenum}
\renewcommand{\theinlineenum}{\enumlabelformat{inlineenum}}
\newenvironment{inlineenum}
 {\setcounter{inlineenum}{0}%
  \renewcommand{\item}{\refstepcounter{inlineenum}{(\theinlineenum)\hspace{\thelabelsep}}}
 }
 {\ignorespacesafterend}

\newcommand\jp{\partial J^{+}(S)}
\newcommand\ep{E^{+}(S)}
\newcommand\dpep{D^{+}\left(E^{+}(S)\right)}
\newcommand\dpjp{D^{+}\left(\partial J^{+}(S)\right)}
\newcommand\hpjp{H^{+}\left(\partial J^{+}(S)\right)}
\newcommand\hpep{H^{+}\left(E^{+}(S)\right)}

\title{The Hawking--Penrose singularity theorem for $C^{1,1}$-Lorentzian metrics}

\author{Melanie Graf\footnote{University of Vienna, Faculty of Mathematics,
melanie.graf@univie.ac.at,
michael.kunzinger@univie.ac.at, roland.steinbauer@univie.ac.at}, \\
James D.E.~Grant\footnote{Department of Mathematics, University of Surrey, j.grant@surrey.ac.uk},\\
Michael Kunzinger${}^*$,\\ 
Roland Steinbauer${}^*$,\\ 
}

\begin{document}

\date{\today}


\maketitle

\begin{abstract}
We show that the Hawking--Penrose singularity theorem, and the generalisation of
this theorem due to Galloway and Senovilla, continue to hold for Lorentzian
metrics that are of $C^{1, 1}$-regularity. We formulate appropriate weak
versions of the strong energy condition and genericity condition for $C^{1,
1}$-metrics, and of $C^0$-trapped submanifolds. By
regularisation, we show that, under these weak conditions, causal geodesics necessarily become non-maximising. This requires a detailed
analysis of the matrix Riccati equation for the approximating metrics, which may
be of independent interest.

\vskip 1em

\medskip


\end{abstract}

\section{Introduction}\label{sec:intro}

The classical singularity theorems of General Relativity show that a Lorentzian manifold that satisfies physically ``sensible'' conditions cannot be geodesically complete. In particular, if one attempts to ``extend'' such a manifold, then one cannot extend with a $C^2$-Lorentzian metric. It is then natural to ask whether one can extend with a lower regularity Lorentzian metric. In certain situations with a large amount of symmetry, one can show that even a low level of regularity cannot be maintained. For example, in recent work, Sbierski~\cite{Sbierski} has shown that the Schwarzschild solution cannot be extended as a continuous Lorentzian metric.

Generally speaking, the singularity theorems of Penrose~\cite{Pen}, Hawking~\cite{HawkingIII} and Hawking--Penrose~\cite{HP} hold for
$C^2$-Lorentzian metrics. In~\cite{penrosec11} and~\cite{hawkingc11}, it has been shown, however, that the theorems of Penrose and Hawking hold for metrics that are $C^{1, 1}$, i.e.\ metrics that are differentiable, with all derivatives locally Lipschitz. Such a level of regularity is of significance to us for a variety of reasons. From a mathematical point of view, such metrics have the following properties:
\begin{enumerate}
\item The Levi-Civita connection is locally Lipschitz. This is, therefore, the lowest regularity where the classical Picard--Lindel\"{o}f theorem gives existence and uniqueness of solutions of the geodesic equations for the metric. Moreover, the solution of the geodesic equation depends continuously (in fact, Lipschitz continuously) on the initial data.
\item The curvature of the metric is well-defined in $L^{\infty}_{\mathrm{loc}}$. In particular, Rademacher's theorem implies that the curvature exists almost-everywhere.
\end{enumerate}
From the point of view of physics, the curvature of a metric being bounded but discontinuous, rather than blowing up, would, via the Einstein field equations, give rise to (or be generated by) a finite jump in the energy-momentum tensor of the matter variables. This scenario is quite acceptable physically, and arises in the classical example of the Oppenheimer--Snyder solution~\cite{OppSny} and the whole class of matched spacetimes (see e.g.~\cite{L,MaSeno}). As such, there are both physical and mathematical motivations for studying the class of
$C^{1, 1}$-metrics.

When one attempts to generalise the proof of the singularity theorems to the case of a $C^{1, 1}$-metric, however, the fact that the curvature tensor is only defined almost-everywhere poses significant problems.%
\footnote{A number of technical obstacles for a proof in the $C^{1, 1}$-case are listed in Sect.~6.1 of the review article~\cite{Seno1}, see also \cite[Sec.\ 8.1]{SenGar}.} %

The standard proof of the singularity theorems relies on the existence of conjugate points (or focal points) along suitable classes of geodesics in the Lorentzian manifold. Such conjugate points are shown to exist by a study of Jacobi fields (or, equivalently, Riccati equations) along these geodesics. However, if the curvature tensor is only defined almost-everywhere, it is quite possible that, since a geodesic curve has measure zero, the curvature may not be defined along any given geodesic, so the Jacobi equation (and, hence, the notion of a conjugate point) is not well-defined along said geodesic. In Riemannian geometry, a standard example of a metric that is $C^{1, 1}$ but not $C^2$ is the metric on a hemisphere joined at the equator to a flat cylinder~\cite{Peters, Penrose:ImpulsiveGravitationalWaves}. This metric has strictly positive curvature on the hemisphere and zero curvature on the cylindrical part, which implies that the curvature is not well-defined on the geodesic that traverses the join between the two regions. A similar phenomenon occurs in Lorentzian geometry in the Oppenheimer--Snyder model, where the curvature tensor is not well-defined along the geodesics that generate the boundary between the interior and exterior regions of the solution. As such, the notion of a Jacobi field is not defined along such geodesics.

The importance of conjugate points (or focal points) in the proof of the singularity theorems is the connection with maximising properties of causal geodesics. In particular, a causal geodesic from a point
stops being maximising if and only if either a) there exists a distinct causal geodesic between the same endpoints of the same length or b) the geodesic encounters a conjugate point.%
\footnote{A similar statement holds for causal geodesics emanating from a submanifold of $M$.}
Given suitable geometrical conditions on the Lorentzian metric (e.g.\ a Ricci curvature bound, a ``convergence condition'' such as the existence of a trapped surface, and a completeness condition), one can use Riccati comparison techniques to show that all causal geodesics of a suitable type will encounter conjugate points, and hence stop being maximising curves between their endpoints. It should perhaps be pointed out, however, that the cut-locus of a point in a Lorentzian manifold is necessarily a closed set, of which conjugate points form a subset of zero measure. Therefore, almost all geodesics stop maximising due to their intersection with another geodesic with the same endpoint of the same length. As such, most causal geodesics will no longer be maximising even \emph{before\/} they encounter their first conjugate point. However, since such an intersection of geodesics is related to the~\emph{global\/} geometry of the manifold, there is no way to estimate (in terms of, say, the curvature) the distance that one must traverse along a given curve before one encounters such an intersection. The power of conjugate points (and focal points) is the fact that they lead to geodesics no longer being maximising \emph{and\/} we can estimate when they occur.

\smallskip
In this paper, we show that the Hawking--Penrose singularity theorem~\cite{HP} can be generalised to $C^{1, 1}$-Lorentzian metrics. The Hawking--Penrose theorem is, perhaps, the most refined of the classical singularity theorems, in the sense that it requires the most delicate analysis of the effects of curvature. As a consequence, the technical issues that arise from the lack of a suitable concept of a ``conjugate point'' are considerably more pronounced when one attempts to generalise the Hawking--Penrose theorem to the $C^{1, 1}$-setting, than they were with the Penrose or Hawking theorems. The most general version of the Hawking--Penrose theorem, which is stated in ``causal'' language, states the following:
\begin{Theorem}\cite[pp.~538]{HP}
\label{C2HPCausalityVersion}
Let $(M, g)$ be a spacetime with $g$ a $C^2$-metric with the following properties:
\begin{enumerate}[label={(C.\roman*)}, ref={C.\roman*}]
\item\label{thm:HPCausal:C2:1} $M$ is chronological, i.e., contains no closed timelike curves;
\item\label{thm:HPCausal:C2:2} Every inextendible causal geodesic in $M$ contains conjugate points;
\item\label{thm:HPCausal:C2:3} There is an achronal set $S$ such that $E^+(S)$ or $E^-(S)$ is compact.
\end{enumerate}
Then $(M, g)$ is causally geodesically incomplete.
\end{Theorem}

Hawking and Penrose also prove the following more ``analytical'' result:
\footnote{In~\cite{HP}, Theorem~\ref{classicalHP} is proved as a Corollary of Theorem~\ref{C2HPCausalityVersion}. Since the bulk of this paper is dedicated to proving the analogue of Theorem~\ref{classicalHP}, we will hereafter refer to Theorem~\ref{classicalHP} as the ``Hawking--Penrose singularity theorem''.}

\begin{Theorem}\cite[Sec.\ 3, Cor.]{HP}
\label{classicalHP}
A spacetime $(M, g)$ with $C^2$-metric that
\begin{enumerate}[noitemsep, label={(A.\arabic*)}, ref={(A.\arabic*)}]
\item\label{1.1:1} is chronological;
\item satisfies the strong energy condition,
\begin{equation}
\label{smoothstrongenergy}
\Ric(X,X)\geq 0\,\,\forall \,\mathrm{causal}\,X\in TM;
\end{equation}
\item\label{1.1:3} satisfies the genericity condition, i.e., along every causal geodesic $\gamma$ there is a point at which
\begin{equation}
\label{genericityOriginal}
\dot{\gamma}^c \dot{\gamma}^d \dot{\gamma}_{\left[a\right.}R_{\left. b \right]cd\left[e\right.}\dot{\gamma}_{\left. f \right]} \neq 0;
\end{equation}
\item\label{1.1:4} contains at least one of the following
\begin{enumerate}[noitemsep, label={(\roman*)}]
\item a compact achronal set without edge,
\item a closed trapped surface or
\item\label{1.1:4c} a point $p$ such that on every past (or every future) null geodesic from $p$ the expansion $\theta $ of the null geodesics from $p$ becomes negative,
\end{enumerate}
\end{enumerate}
cannot be causally geodesically complete.
\end{Theorem}

For $C^2$-metrics, Theorem~\ref{classicalHP} is proved as a corollary of Theorem~\ref{C2HPCausalityVersion}. In particular, the genericity condition~\eqref{genericityOriginal} along with strong energy condition~\eqref{smoothstrongenergy} are used, in conjunction with a matrix Riccati equation for the second fundamental form of a geodesic congruence, to show that any of the conditions~\ref{1.1:4} imply that every inextendible causal geodesic in $M$ contains conjugate points, and that Condition~(\ref{thm:HPCausal:C2:3}) of Theorem~\ref{C2HPCausalityVersion} holds. Therefore, the conditions of Theorem~\ref{classicalHP} imply those of Theorem~\ref{C2HPCausalityVersion}.

In the $C^{1, 1}$-case, which we study in this paper, the logical structure of the argument is very similar. We first prove an appropriate version of Theorem~\ref{C2HPCausalityVersion} for $C^{1, 1}$-metrics. To this end, we first note that Condition~(\ref{thm:HPCausal:C2:2}) in Theorem~\ref{C2HPCausalityVersion} explicitly depends on the concept of a conjugate point, and so cannot be directly generalised to the case of $C^{1, 1}$-metrics. However, an inspection of the proof of the Hawking--Penrose theorem shows that, rather than Condition~(\ref{thm:HPCausal:C2:2}), the property that is actually required for their result is the following:

\begin{enumerate}[label={(\ref{thm:HPCausal:C2:2}{\color{red}${}^{\prime}$})}, ref={\ref{thm:HPCausal:C2:2}${}^{\prime}$}]
\item
\label{thm:HPcausal:C2:2:minimising} Every inextendible causal geodesic in $M$ stops being maximising;
\end{enumerate}
One of our fundamental results is, therefore, Theorem~\ref{C11HPCausalitybit}, which states that, with minor modifications, Theorem~\ref{C2HPCausalityVersion}, with Condition~(\ref{thm:HPCausal:C2:2}) replaced with Condition~(\ref{thm:HPcausal:C2:2:minimising}) continues to hold if the metric $g$ is assumed to be $C^{1, 1}$. The web of causality results required in the proof of Theorem~\ref{C2HPCausalityVersion}, generalised to the $C^{1, 1}$-setting, is summarised in Appendix~\ref{app:C11causality}.

In the $C^{1, 1}$-case, however, the step from Theorem~\ref{C2HPCausalityVersion} to Theorem~\ref{classicalHP} is considerably more complicated. We show that appropriate versions of the curvature conditions~\eqref{smoothstrongenergy} and~\eqref{genericityOriginal} lead to causal geodesics becoming non-maximising between their endpoints. We prove this result by studying appropriate smooth approximations $g_{\eps}$ to the $C^{1, 1}$-metric $g$, where the $g_{\eps}$ satisfy appropriate weakened versions of~\eqref{smoothstrongenergy} and~\eqref{genericityOriginal}. By a refined analysis of the matrix Riccati equation along geodesics with respect to the
$g_{\eps}$-metrics, we are able to show that $g_{\eps}$-causal geodesics develop conjugate points,%
\footnote{Note that the metrics $g_{\eps}$ are smooth, so the classical notion of a conjugate point is well-defined.} %
and, hence, are non-maximising. From this, we argue that $g$-causal geodesics also become non-maximising. At this point, our main results, Theorem~\ref{HPinC1,1} and~Theorem~\ref{ArbitrCodiminC1,1} follow from Theorem~\ref{C11HPCausalitybit}.

The techniques that we develop in going from Theorem~\ref{C11HPCausalitybit} to Theorem~\ref{HPinC1,1} and~Theorem~\ref{ArbitrCodiminC1,1} are the main technical developments in this paper. In particular, the estimates that we develop in Sections~\ref{reg_sec} and~\ref{sec:conjugate_points} are new,%
\footnote{To the best of our knowledge.}
and may well be of independent interest.%
\footnote{In particular, these are~\emph{not\/} estimates that follow from the standard Rauch comparison theorem for Jacobi fields.}

\bigskip

We conclude this introduction by fixing our notation and conventions as well as
introducing an improved version of the smooth Hawking--Penrose theorem that we
will also deal with during this work.

All manifolds will be denoted by $M$ and assumed to be smooth, Hausdorff, second
countable, $n$-dimensional (with $n\geq 3$), and connected. On such $M$ we will
consider Lorentzian metrics $g$ of regularity of at least $C^{1,1}$ and
signature $(-,+\dots,+)$ with Levi-Civita connection $\nabla$ and with a time
orientation fixed by a continuous vector field. We say a curve $\gamma : I\to M$
from some interval $I\sse \R$ to $M$ is timelike (causal, null, future or past
directed) if it is locally Lipschitz and $\dot{\gamma}(t)$,
which exists almost everywhere by Rademacher's theorem, is timelike (causal,
null, future or past directed) almost everywhere. Following standard notation,
for $p,q\in M$ we write $p\ll q$ if there exists a future directed timelike
curve from $p$ to $q$ (and $p\leq q$ if there exists a future directed causal
curve from $p$ to $q$ or $p=q$) and set $I^+(A):=\{q\in M:\, p\ll q\
\mathrm{for\,some}\,p\in A\}$ and $J^+(A):=\{q\in M:\, p\leq q\
\mathrm{for\,some}\,p\in A\}$. We note that we require causal (timelike, \dots)
curves to be Lipschitz, whereas other standard sources use piecewise $C^1$
curves instead (see, e.g., \cite{HE}, \cite{ON83}). However, as was shown in
\cite[Thm.~7]{M}, \cite[Cor.~3.10]{KSSV}, this has no impact on the relations
$\ll$ and $\leq $ for $C^{1,1}$-metrics. We call a $C^{1,1}$-spacetime $(M,g)$
globally hyperbolic if it is causal (i.e., contains no closed causal curves) and
$J(p,q):=J^+(p)\cap J^-(q)$ is compact for all $p,q\in M$. We further define the
Riemann curvature tensor%
\footnote{Note that
we follow the convention of~\cite{HE} for the curvature tensor, which is the
opposite of that employed in~\cite{ON83,hawkingc11,penrosec11}.} %
by $R(X,Y)Z=[\nabla_X,\nabla_Y]Z-\nabla_{[X,Y]}Z$ and the Ricci tensor by
$\Ric(X,Y)=\sum_{i=1}^n\langle E_i,E_i\rangle\langle R(E_i,X)Y,E_i\rangle$,
which in case of $g$ being $C^{1,1}$ are $L^\infty_{\mbox{\scriptsize
loc}}$-tensor fields. Here and in the following $(E_i)_{i=1}^n$ will denote
(local) orthonormal frame fields and $(e_i)_{i=1}^n$ will denote orthonormal frames in individual
tangent spaces $T_pM$. Generally we will consider embedded submanifolds $S$ of
codimension $m$. We define the second fundamental form by
$\mathrm{II}(V, W) := \mbox{nor}(\nabla_V W)$ for all $V, W$ tangent to $S$ and the shape operator derived from a
normal unit field $\nu$ by $S_{\nu}(X)=\nabla_X\nu$. For any tangent vector $v \in T_p M$ we denote by $\gamma_v$ the geodesic with $\gamma_v(0) = p$, $\dot{\gamma}_v(0) = v$. Throughout, a codimension $2$ submanifold of $M$ will be referred to as a ``surface''.

Condition~\ref{1.1:4}(ii) of Theorem~\ref{classicalHP} has been generalized
in~\cite{GS} to
include trapped submanifolds of arbitrary co-dimension $m$ ($1< m< n$) by
adding an additional curvature assumption, which in the classical case
$m=2$ automatically follows from the energy condition.
For a precise formulation let $S$ be a (smooth) spacelike $(n-m)$-dimensional
submanifold and let $e_1(q),\dots,e_{n-m}(q)$ be an orthonormal basis for $T_q S$,
smoothly varying with $q$ in a neighbourhood (in $S$) of $p \in S$.
For a geodesic $\gamma$ starting at $p$ let $E_1,\dots,E_{n-m}$ denote the parallel translates of
$e_1(p),\dots,e_{n-m}(p)$ along $\gamma$. Let $H_S :=
\frac{1}{n-m}\sum_{i=1}^{n-m}\mathrm{II}(e_i,e_i)$ denote the mean curvature
vector field of $S$, and let $\mathbf{k}_S(v):=g(H,v)$ be the
convergence of $v\in TM|_S$. Now a closed spacelike submanifold
$S$ is called \emph{(future) trapped\/} if for any future-directed
null vector $\nu \in T
S^\perp$ the convergence $\mathbf{k}_S(\nu)$ is positive. This is equivalent to the mean curvature vector field $H_S$ being past pointing timelike on all of $S$.
With this definition one has the following extension of the classical Hawking--Penrose theorem
(\cite[Thm.\ 3]{GS}).

\begin{Theorem}\label{classicalArbitrCodim}
 A spacetime $(M,g)$ with $C^2$-metric satisfying conditions~\ref{1.1:1}--\ref{1.1:3} of Theorem~\ref{classicalHP} and
 \begin{enumerate}
  \item [(A.4)]
  \begin{enumerate}
  \item [(iv)] contains a spacelike (future) trapped submanifold
 $S$ of co-dimension $2 < m < n$ such that additionally
 \begin{equation}
  \sum_{i=1}^{n-m} \langle R(E_i,\dot{\gamma})\dot{\gamma},E_i \rangle
  \geq 0
 \end{equation}
 for any future directed null geodesic with $\dot{\gamma}(0)$ orthogonal
 to $S$,
 \end{enumerate}
 \end{enumerate}
cannot be causally geodesically complete.
\end{Theorem}

\noindent{} In Section \ref{sec:proof}, this result will also be shown to hold in the $C^{1, 1}$-setting.

\medskip

This paper is organised in the following way. In
Section~\ref{sec:mainresult}, we first define the appropriate weak notions of
curvature conditions on Lorentzian metrics and convergence conditions on
$C^0$-submanifolds that are required for our study of metrics that are $C^{1,
1}$. We then state our main results, Theorems~\ref{HPinC1,1}
and~\ref{ArbitrCodiminC1,1}, which are the analogues of the Hawking--Penrose
Theorem~\ref{classicalHP} and its generalisation,
Theorem~\ref{classicalArbitrCodim}, to the $C^{1, 1}$-case. The remainder of the
paper is concerned with the proof of these results. In Section~\ref{reg_sec}, we
consider the regularisation of the $C^{1, 1}$-metric and, in particular, study
the effect of smoothing on the curvature and genericity condition. In
Section~\ref{sec:conjugate_points}, we develop estimates for matrix Riccati
equations that allow us to show that geodesics with respect to the smooth
approximating metrics must develop conjugate (or focal) points. As mentioned
previously, the estimates obtained in Sections~\ref{sec:conjugate_points} are,
perhaps, the main technical advance in this paper, and may be of independent
interest in their own right. The results of Section~\ref{sec:conjugate_points}
are used in Section~\ref{sec:max_geods} to yield Theorems~\ref{timelikenotmax}
and~\ref{3.4}, which show that, under our curvature and genericity assumptions,
causal geodesics will not remain maximising. In Section~\ref{sec:trapped}, we
show that if $S$ is a submanifold of $M$ satisfying any one of the conditions
(A.4) of Theorems~\ref{HPinC1,1} and~\ref{ArbitrCodiminC1,1}, then $E^+(S)$ is
compact, i.e., the submanifold is a trapped set. Finally, in Section~\ref{sec:proof},
we first show, using results summarised in Appendix~\ref{app:C11causality}, that
Theorem~\ref{C11HPCausalitybit}, the analogue of the ``causal'' version of the
Hawking--Penrose Theorem (Theorem~\ref{C2HPCausalityVersion}), holds in the
$C^{1, 1}$-setting. The results from Sections~\ref{reg_sec}--\ref{sec:trapped}
then quickly yield the main result Theorems~\ref{HPinC1,1}
and~\ref{ArbitrCodiminC1,1}, i.e. the ``analytical'' version of the
Hawking--Penrose theorem.

\section{The main result}\label{sec:mainresult}
The aim of this paper is to generalise Theorems~\ref{classicalHP} and~\ref{classicalArbitrCodim} to $C^{1, 1}$-metrics. Since not all of the conditions in these theorems are well-defined at this lower level of regularity, we begin by discussing the alternative formulations that we will use in the $C^{1,1}$- case.

By the \emph{strong energy condition\/} or \emph{causal convergence condition},
we shall mean that
\begin{align}\label{eq:sec}
\Ric(X,X)\geq 0 \quad\mbox{for all Lipschitz
	continuous causal local vector fields $X$.}
\end{align}
We will also speak of the timelike (or null) convergence
condition if \eqref{eq:sec} is only supposed to hold for all Lipschitz
continuous timelike (or null) local vector fields $X$.

\begin{remark}
\label{remark:null}
This condition is natural in the $C^{1, 1}$-context and has been successfully used in the proofs of other singularity theorems in this regularity (cf.~\cite[Rem.\ 1.2(i)]{hawkingc11} and~\cite[Rem.\ 1.2(i)]{penrosec11}).
Note that the Lipschitz condition is only relevant in the null case.
Contrary to the situation with a timelike vector, which can clearly be extended to a smooth timelike local vector field,
it is, in general, \emph{not\/} possible to extend a given null vector to a \emph{smooth\/} null local vector field.
Indeed, parallel transporting a given null vector at a given point along radial geodesics emanating from that point results in a null vector field that is only Lipschitz continuous.
It is possible that, with a $C^{1, 1}$-Lorentzian metric, one can extend a given null vector to a $C^{1, 1}$ null local vector field, and the condition for our results may be weakened to requiring~\eqref{eq:sec} to hold for all $C^{1, 1}$ causal local vector fields $X$. However, since we will explicitly use a null vector field obtained by parallel transport (and, hence, Lipschitz) in the proof of Lemma~\ref{lem:2.3}, we have not investigated this possibility.
For simplicity, we also refrain from refining condition~\eqref{eq:sec} to apply to local smooth
timelike and Lipschitz null vector fields, although this would be possible throughout.
\end{remark}

Looking at the classical proof of Theorem~\ref{classicalHP}, one finds that it is not the genericity condition itself that plays a role, but rather a derived condition on the tidal force operator
along causal geodesics $\gamma$. The required condition is that there exists $t_0$ such that the operator
\begin{equation}\label{TidalGenericitySmooth}
R \colon (\dot{\gamma}(t_0))^\perp \to (\dot{\gamma}(t_0))^\perp,\,\,\quad v \mapsto R(v,\dot{\gamma})\dot{\gamma}
\end{equation}
is not identically zero. (The fact that this condition follows from the genericity condition~\eqref{genericityOriginal} can be found in, e.g.,~\cite[Cor.~9.1.1]{Krie}.)
Thus, we will henceforth refer to~\eqref{TidalGenericitySmooth} as the genericity condition, 
which we now formulate for $C^{1, 1}$-metrics, and which reproduces~\eqref{TidalGenericitySmooth} in the smooth case, as we shall see below (Lemma \ref{lem:2.2}).

\begin{definition}\label{genericity_c11}
	Let $g \in C^{1,1}$ be a Lorentzian metric on $M$, and let $\gamma \colon I \to M$ be a
	causal geodesic for $g$. Then we say that the \emph{genericity condition\/}
	holds along $\gamma$ if there exists some $t_0 \in I$ and a neighbourhood $U$ of
	$\gamma(t_0)$, as well as continuous vector fields $X$ and $V$ on $U$
	such that $X(\gamma(t)) = \dot{\gamma}(t)$ and $V(\gamma(t)) \in (\dot{\gamma}(t))^\perp$
	for all $t \in I$ with $\gamma(t) \in U$, and there exists some $c>0$ such that
	\begin{equation} \label{posgenericity}
	\langle R(V,X)X,V\rangle > c
	\end{equation} 
	in $L^{\infty}(U)$. In this case, we say that the genericity condition is satisfied for $\gamma$ at $t_0\in I$.
\end{definition}

Regarding the \emph{initial conditions}~\ref{1.1:4}, we first remark that the
definition of an ``achronal set without edge'' and of a ``smooth (or at least
$C^2$-) future trapped submanifold'' for $C^{1, 1}$-metrics can be carried across
unchanged from the smooth case since the mean curvature is still Lip\-schitz
continuous. We will however wish to generalise the notion of a future trapped
submanifold slightly to allow us to use $C^0$-submanifolds. We say that a(n at
least $C^2$) submanifold $\widetilde{S}$
is a \emph{future support submanifold for a $C^0$-submanifold $S$ at $q\in S$\/} if
$\dim (\widetilde{S})=\dim S$, $q\in \widetilde{S}$, and $\widetilde{S}$ is locally to the future of $S$
near $q$, i.e. there exists a neighbourhood $U$ of $q$ in $M$ such that
$\widetilde{S}\cap U \subset J^+(S,U)$. Using such future support submanifolds we define \emph{past pointing timelike mean curvature at $q\in S$\/} by requiring the existence of a future support submanifold with past-pointing timelike mean curvature at $q$ (see, for instance,~\cite{AGH}).

\smallskip
This leads to the following definition of a future trapped submanifold of $M$
(which reduces to the usual one if $S$ is at least $C^2$).

\begin{definition}\label{def:trappedsubmf1}
	A closed ($C^0$-) submanifold $S$ of codimension $m$ ($1 \le m<n$) is called
	\emph{future trapped\/} if, for any $p \in S$, there exists a neighbourhood $U_p$ of $p$ such that $S \cap U_p$ is achronal in $U_p$ and $S$ has past-pointing timelike mean
	curvature at all of its points (in the sense of support submanifolds).
\end{definition}

Similarly, to replace the point condition~\ref{1.1:4}\ref{1.1:4c}
in~Theorem~\ref{classicalHP}, we define a \emph{(future) trapped point\/} as
follows:

\begin{definition}\label{def:trappedpt1}
	We say that a point $p$ is \emph{future trapped} if, for any future-pointing null vector $\nu \in T_pM$, there exists a $t$ such that there exists a spacelike $C^2$-surface $\widetilde{S}\subset J^+(p)$ with $\gamma_\nu(t)\in \widetilde{S}$ and $\mathbf{k}_{\widetilde{S}}(\dot{\gamma}_\nu(t))>0$.
\end{definition}

While it is perhaps not immediately obvious that this provides a good generalisation of the usual condition, one can show that for smooth metrics there is a very clear relationship between the expansion $\theta (t)$ along a geodesic $\gamma$ defined in terms of Jacobi tensor classes (cf.~Lemma~\ref{boxing}) and the shape operator $S_{\dot{\gamma}}(t)$ derived from $\dot{\gamma}$ for the submanifold $S_t:=\exp_p(t\,V)$, where $V$ is the set of all (properly normalised) null vectors contained in some neighbourhood of $\dot{\gamma}(0)$ (see section \ref{sec:trappedpt} for details). Our definition then provides a $C^{1,1}$-generalisation of the trace of such a shape operator becoming negative.

With these definitions we will prove the following generalisation of Theorem~\ref{classicalHP}:

\begin{Theorem}[Hawking--Penrose for $C^{1,1}$-metrics] \label{HPinC1,1}
	Let $(M, g)$ be a spacetime with a $C^{1,1}$-metric. If $M$
	\begin{enumerate}[label={(A.\arabic*)}, ref={(A.\arabic*)}]
		\item\label{HPinC1,1:1} is causal;
		\item\label{HPinC1,1:2} satisfies the strong energy condition~\eqref{eq:sec};
		\item\label{HPinC1,1:3} satisfies the genericity condition along any inextendible causal geodesic (Definition~\ref{genericity_c11});
		\item\label{HPinC1,1:4} contains at least one of the following
		\begin{enumerate}[ref={(A.4.\roman*)}]
			\item\label{HPinC1,1:4i} a compact achronal set without edge;
			\item\label{HPinC1,1:4ii} a closed future trapped ($C^0$-)surface (Definition~\ref{def:trappedsubmf1});
			\item\label{HPinC1,1:4iii} a future trapped point (Definition~\ref{def:trappedpt1}),
		\end{enumerate}
	\end{enumerate}
	then it cannot be causally geodesically complete.
\end{Theorem}

Note that the $C^{1, 1}$-version requires that $(M, g)$ be \emph{causal\/} rather than \emph{chronological\/} since, contrary to the smooth case, the other conditions that we impose do not exclude the existence of closed null curves. The problem will be evident in the proof of Theorem~\ref{3.4}, where we will use approximations to show that no inextendible null geodesic can be globally maximising, and our argument breaks down for closed null curves.

\smallskip
Finally, we will also prove a $C^{1,1}$-generalization of Theorem~\ref{classicalArbitrCodim}.

\begin{Theorem}\label{ArbitrCodiminC1,1}
	Let $(M, g)$ be a spacetime with a $C^{1, 1}$-metric that satisfies conditions~\ref{HPinC1,1:1} to~\ref{HPinC1,1:3} of Theorem~\ref{HPinC1,1} and
 \begin{enumerate}[label={(A.\arabic*)}, ref={(A.\arabic*)}]
  \setcounter{enumi}{3}
   \item
   \begin{enumerate}[ref={(A.4.\roman*)}]
   \setcounter{enumii}{3}
   \item\label{HPinC1,1:4iv}
	contains a (future) trapped $C^0$-submanifold (Definition~\ref{def:trappedsubmf1}) of co-dimension $2 < m < n$  such that the support submanifolds $\tilde S$
additionally satisfy the following:
For any future directed null geodesic $\gamma$ starting orthogonally to $\widetilde{S}$ there exist $b>\frac{1}{\mathbf{k}_{\widetilde{S}}(\dot{\gamma}(0))}$,  a neighbourhood $U$ of $\gamma|_{\left[0,b\right]} $, and continuous
	extensions $\bar{E}
	_1,\dots \bar{E}_{n-m}$ and $\bar{N}$ of $E_1,\dots E_{n-m}$ (for $\widetilde{S}$) and $N:=\dot{\gamma}$,
	respectively, to $U$ such that
	\begin{equation}
	\sum_{i=1}^{n-m} \langle R(\bar{E}_i,\bar{N})\bar{N},\bar{E}_i \rangle \geq
	0\qquad \mathrm{a.e.\;on\;}U .
	\end{equation}
	\end{enumerate}
	\end{enumerate}
	Then $M$ contains an incomplete causal geodesic.
\end{Theorem}
\section{Regularisation results}
\label{reg_sec}
In this section we establish a number of auxiliary results pertaining to regularisations of
$C^{1,1}$-metrics, as well as the corresponding curvature quantities and geodesics. Our approach
rests on the causality-respecting regularisation procedure introduced by Chru\'sciel and Grant
in~\cite{CG}. In its formulation, we shall employ the
following notation (cf.~\cite[Sec.\ 3.8.2]{Minguzzi08thecausal},~\cite[Sec.\ 1.2]{CG}):
Given Lorentzian metrics $g_1$,
$g_2$, we say that $g_2$ has \emph{strictly wider light cones} than $g_1$, denoted by
$g_1\prec g_2$, if for any tangent vector $X\not=0,\ g_1(X,X)\le 0$ implies that $g_2(X,X)<0$.
Thus any $g_1$-causal vector is timelike for $g_2$.
Then~\cite[Prop.\ 1.2]{CG} (cf.\ also~\cite[Prop.\ 2.5]{KSSV}) gives:

\begin{Proposition}\label{CGapprox} Let $(M,g)$ be a $C^0$-spacetime
and let $h$ be some smooth
background Riemannian metric on $M$. Then for any $\eps>0$, there exist smooth
Lorentzian metrics $\check g_\eps$ and $\hat g_\eps$ on $M$ such that
for all $0<\eps<\eps'$, $\check g_{\eps'} \prec \check g_\eps
\prec g \prec \hat g_\eps \prec \hat g_{\eps'}$, and $d_h(\check g_\eps,g) + d_h(\hat g_\eps,g)<\eps$,
where
\begin{equation}\label{CGdh}
d_h(g_1,g_2) := \sup_{p\in M,0\not=X,Y\in T_pM} \frac{|g_1(X,Y)-g_2(X,Y)|}{\|X\|_h
\|Y\|_h}.
\end{equation}
Moreover, $\hat g_\eps(p)$ and $\check g_\eps(p)$ depend smoothly on $(\eps,p)\in \R^+\times M$, and if
$g\in C^{1,1}$ then, letting $g_\eps$ be either $\check g_\eps$ or $\hat g_\eps$,
we additionally have
\begin{itemize}
 \item[(i)] $g_\eps$ converges to $g$ in the $C^1$-topology as $\eps\to 0$, and
 \item[(ii)] the second derivatives of $g_\eps$ are bounded, uniformly in $\eps$, on compact sets.
 \end{itemize}
\end{Proposition}
\noindent Curvature quantities for $g_\eps$-metrics will be denoted by a subscript, as in $R_\eps$ or
$\Ric_\eps$.

Next we recall the consequences of the strong energy condition~\eqref{eq:sec}
provided by~\cite[Lemma 3.2]{hawkingc11} and~\cite[Lemma 2.4]{penrosec11} for
nets $(g_\eps)_{\eps>0}$ (with $g_\eps=\check{g}_\eps$ or $g_\eps = \hat g_\eps$)
of approximating smooth metrics.

\begin{Lemma}\label{approxlemma} Let $M$ be a smooth manifold with a
$C^{1,1}$-Lorentzian metric $g$ and smooth Riemannian background metrics
$h$, $\tilde h$ on $M$ and $TM$, respectively. Let $K\comp M$ and let $C$,
$\delta > 0$. Then we have:
\begin{enumerate}
 \item[(i)]
  If $\Ric(Y,Y)\ge 0$ for every $g$-timelike smooth local vector field $Y$, then
\begin{equation}\label{suffest}
\begin{split}
  &\forall \kappa<0\ \exists \eps_0>0\ \forall \eps<\eps_0\
  \forall X\in TM|_K \text{ with }\ g(X,X)\le \kappa \\
  & \text{ and } \|X\|_h \leq C:
  \ \Ric_\eps(X,X) > -\delta.
  \end{split}
\end{equation}
 \item[(ii)] If $\Ric(Y,Y)\ge 0$ for every Lipschitz-continuous $g$-null local
  vector field $Y$, then
  \begin{align}\nonumber
  &\exists \eta>0\ \exists \eps_0>0 \ \forall \eps<\eps_0:\
   \mbox{if $p\in K$, $X\in T_pM$ with $\|X\|_h \le C$}\\
  & \mbox{and $\exists Y_0\in TM|_K$, $g$-null with $d_{\tilde
    h}(X,Y_0) \le \eta$ and $\|Y_0\|_h\le C$}:\\
  &\Ric_\eps(X,X) > -\delta.\nonumber
\end{align}
\end{enumerate}
\end{Lemma}
For later use, we also record the following result, cf.\ e.g.\ the proof
of~\cite[Prop.\ 4.3]{hawkingc11}:
\begin{Lemma}\label{d_conv} Let $(M,g)$ be a globally hyperbolic $C^{1,1}$-spacetime
and let $p$, $q\in M$. Denote by $d$ and $d_{\gec}$ the
time-separation functions with respect to $g$ and $\gec$, respectively.
Then, we have
\[
 d_{\gec}(p,q) \to d(p,q) \qquad (\eps\to 0).
\]
\end{Lemma}

The following basic Friedrichs-type Lemma collects some general convergence properties that will
be used repeatedly in subsequent sections.

\begin{Lemma}\label{reg1} Let $a\in \linfloc(\R^n)$, $f\in C^0(\R^n)$, $b_\eps\in
C^0(\R^n)$ ($\eps>0$), and $b_\eps\to b$ locally uniformly for $\eps\to 0$. Let
$\rho\in \D(\R^n)$
	be a standard mollifier. Then
	\begin{itemize}
		\item[(i)] $(a\cdot f \cdot b)*\rho_\eps - (a*\rho_\eps)\cdot
(f*\rho_\eps)\cdot b_\eps \to 0$ ($\eps\to 0$) locally uniformly.
		\item[(ii)] If $\rho$ is non-negative and $a\cdot f\cdot b\ge c\in \R$ then
		$$
		\forall \tilde c< c\ \forall K\comp \R^n \ \exists \eps_0 \ \forall \eps<\eps_0: (a*\rho_\eps)\cdot (f*\rho_\eps)\cdot b_\eps > \tilde c \quad \text{ on } K.
		$$
	\end{itemize}
\end{Lemma}
\begin{proof} (i) We have
\begin{equation*}
\begin{split}
(a\cdot f \cdot b)*\rho_\eps - (a*\rho_\eps)\cdot (f*\rho_\eps)\cdot b_\eps =
 (a\, \cdot \, & f \cdot b)*\rho_\eps - (a\cdot f)*\rho_\eps \cdot b*\rho_\eps\\
 &+ (a\cdot f)*\rho_\eps \cdot b*\rho_\eps - (a*\rho_\eps)\cdot (f*\rho_\eps)\cdot b_\eps.
 \end{split}
\end{equation*}
Here, both the first and the second term on the right hand side go to $0$ locally uniformly by a variant of the Friedrichs Lemma
(cf.\ the proof of~\cite[Lemma 3.2]{hawkingc11}).

(ii) Since $(a\cdot f \cdot b)*\rho_\eps \ge c$, the claim follows from (i).
\end{proof}

A convenient consequence of the previous Lemma concerns basic properties of
curvature quantities associated to a $C^{1,1}$-metric $g$: Arguing in a local
chart, Lemma~\ref{reg1} shows that if $g_\eps$ is as in Proposition~\ref{CGapprox},
then $R_\eps - R*\rho_\eps\to 0$ locally uniformly
(cf.\ (5) in~\cite{hawkingc11}). Since, moreover, $R*\rho_\eps\to R$ in any
$L^p_{\mbox{\scriptsize loc}}$ ($1\le p<\infty$), all the usual symmetry
properties of the Riemann tensor for smooth metrics carry over to $R$ pointwise
a.e.

Next we introduce some notation to deal with timelike and null geodesics
simultaneously. Suppose that $\gamma$ is a causal geodesic in a
$C^{1,1}$-spacetime $(M,g)$. As is common in the smooth case (see e.g.\
\cite[Sec.\ 4.6.3]{Krie}) we consider the
quotient space $[\dot\gamma(t)]^\perp:=(\dot\gamma(t))^\perp/\R\dot\gamma(t)$, i.e. vectors $v, w \in \left(\dot{\gamma}(t)\right)^{\perp}$ are equivalent if there exists $\alpha \in \R$ such that $v = w + \alpha \dot{\gamma}(t)$.
In the case where $\gamma$ is null, $[\dot\gamma(t)]^\perp$ is an
$(n-2)$-dimensional subspace of
$(\dot\gamma(t))^\perp$. When $\gamma$ is timelike, $[\dot\gamma(t)]^\perp$ coincides with $(\dot{\gamma}(t))^{\perp}$.
In order to
enable a unified notation we will henceforth denote the dimension of
$[\dot\gamma(t)]^\perp$ by $d$, i.e. $d=n-2$ in the null
case and $d=n-1$ in the timelike case. Also
we set $[\dot\gamma]^\perp=\bigcup_t[\dot\gamma(t)]^\perp$.
Every normal tensor field $A$ along $\gamma$ then induces a unique tensor
class $[A]$ along $\gamma$ and the induced covariant derivative
$\nabla_{\dot\gamma}$ is well-defined for tensor classes
and denoted by $[\dot A]=[\nabla_{\dot\gamma}A]$.
The metric $g|_{[\dot\gamma]^\perp}$ is positive definite in both the null and the timelike case.
Also
recall that, for smooth metrics, the curvature (or tidal force) operator $[R](t):\,
[\dot\gamma(t)]^\perp\to [\dot\gamma(t)]^\perp$,
$[v]\mapsto[R(v,\dot\gamma(t))\dot\gamma(t)]$ is well-defined since
$R(\dot\gamma,\dot\gamma)\dot\gamma=0$.
\medskip

Before we proceed to construct suitable frames for the approximating curvature
operators $[R_\eps](t)$, we will show that for the case of a $C^2$-Lorentzian metric
our definition of genericity (Definition \ref{genericity_c11}) is
equivalent to the classical one, i.e., 
\eqref{TidalGenericitySmooth} if
the strong energy condition~\eqref{eq:sec} holds. Clearly, \eqref{posgenericity} implies
\eqref{TidalGenericitySmooth}. For the converse, we have:
\begin{Lemma} \label{lem:2.2}
Let $g\in C^{2}$ be a Lorentzian metric on $M$, and let $\gamma \colon I\to M$ be a
causal geodesic for $g$. Suppose that the genericity condition \eqref{TidalGenericitySmooth}  is satisfied for $\gamma$ at $t_0\in I$. If the strong energy condition~\eqref{eq:sec} holds
then there exist a neighbourhood $U$ of $\gamma(t_0)$, as well as
Lipschitz vector fields $X$ and $V$ on $U$ such that $X(\gamma(t)) =
\dot \gamma(t)$ and $V(\gamma(t))\in \left( \dot{\gamma}(t) \right)^\perp$ for all $t\in I$
with $\gamma(t)\in U$, and there exists some $c>0$ such that $ \langle
R(V,X)X,V\rangle >c$ on $U$. 
\end{Lemma}
\begin{proof}
We assume that~\eqref{TidalGenericitySmooth} holds at $t_0$. Let $e_1,\dots, e_n$ be orthonormal
vectors at $\gamma(t_0)$ (with $e_1,\dots,e_{n-1}$ spacelike and $e_n$ timelike) such that $\dot \gamma(t_0) =(e_{n-1}+e_n)$ if $\dot\gamma(t_0)$ is null or  $\dot \gamma(t_0) = e_n$ if $\dot \gamma(t_0)$ is timelike, respectively. Then $\Ric(\dot\gamma(t_0),\dot\gamma(t_0))=\sum_{i=1}^{k} \langle R(e_i,\dot\gamma(t_0))\dot\gamma(t_0),e_i \rangle\ge 0$, where $k=n-2$ in the null case
and $k=n-1$ in the timelike case. Due to \eqref{TidalGenericitySmooth}, at least one of the 
summands, say $\langle R(e_j,\dot\gamma(t_0))\dot\gamma(t_0),e_j \rangle$  has to be strictly positive. By continuity, extending $e_j$ and $\dot\gamma(t_0)$ to a neighbourhood $U$ of 
$\gamma(t_0)$ (e.g.\ by parallel transport) provides the desired vector fields $V$ and $X$
such that \eqref{posgenericity} is satisfied.
\end{proof}

The next step is to use the $C^{1,1}$-genericity condition to derive a lower
bound on the tidal force operator for approximating metrics along approximating
causal geodesics.

\begin{Lemma}\label{lem:2.3}
 Let $g\in C^{1,1}$ be a Lorentzian metric on $M$ such that the strong energy
 condition is satisfied, and let $\gamma \colon I\to M$ be
 a causal geodesic for $g$. Suppose that the genericity condition is
 satisfied for $\gamma$ at $t_0\in I$.
 Then there exist constants $r>0$, $c>0$, and $C>0$
 such that the following holds:
 Let $g_\eps = \check g_\eps$ or $g_\eps = \hat g_\eps$, and let $\gamma_\eps$
 be $g_\eps$-geodesics of the same causal character w.r.t.\ $g_\eps$ as that
 of $\gamma$
 w.r.t.\ $g$. Assume that $\gamma_\eps$ converges to $\gamma$ in $C^1(I)$
 and for each $\eps$, let
\[
[R_\eps](t) := [R_\eps(\,.\,,\dot \gamma_\eps(t))\dot \gamma_\eps(t)] \colon [\dot \gamma_\eps(t)]^\perp \to [\dot \gamma_\eps(t)]^\perp.
\]
 Then there exists $\eps_0>0$ such that, for each $\eps\in (0,\eps_0)$
 there is a smooth parallel orthonormal frame $[E_1^\eps](t),\dots,$
 $[E_{d}^\eps](t)$ for $[\dot\gamma_\eps]^\perp$ such that
 \begin{align}
  [R_\eps](t) > \diag(c,-C,\dots,-C)\ \mbox{on $[t_0-r,t_0+r]$}
 \end{align}
 in terms of this frame.%
\footnote{Here and below, for $d \times d$ matrices $A, B$, we write $A > B$ if the matrix $A-B$ is positive definite.}
\end{Lemma}

\begin{remark}\label{rem_tl_null}
As the proof will show, the conclusion of Lemma~\ref{lem:2.3} remains valid if,
for $\gamma$ timelike resp.\ null, also the strong energy resp.\ genericity
condition are assumed to hold only for the timelike resp.\ null case.

Moreover, since in all the following results the strong energy condition only
enters via Lemmas~\ref{approxlemma} and~\ref{lem:2.3}, the
claim in the final sentence of Remark~\ref{remark:null} indeed holds. 
\end{remark}

\begin{proof}[Proof of Lemma~\ref{lem:2.3}]
As the claim is local, we may assume that $M=\R^n$.
We use the notation of Definition~\ref{genericity_c11}, and may clearly set
$t_0=0$.
Additionally we may assume that $\gamma $ is parametrised to unit speed
(if $\gamma$ is timelike) or such that $\dot \gamma(0) = e_{n-1}+e_n$ for two
orthonormal vectors $e_{n-1}, e_n$ with $e_n$ timelike (if $\gamma$ is null).
Setting $e_n:=\dot\gamma(0)$ in the timelike case, by shrinking $U$ and $c$
we may assume that $U$ is totally normal (\cite[Sec.\ 4]{KSS})
and relatively compact and replace $X$ 
by the parallel transport (radially outward from $\gamma(0)$)
of $e_{n-1}+e_n$ in the null case, respectively $e_n$ in the timelike
case. 

We now briefly distinguish the timelike and the null case, first assuming that
$\gamma $ is null. We then replace $V$ by
the vector field obtained by transporting $V(\gamma(t_0))$ outwards from
$\gamma(t_0)$ along radial geodesics.
Then by possibly shrinking $U$ and $c$
we still retain the genericity estimate~\eqref{posgenericity} for $X$ and $V$.
By construction, the new $V$ is either proportional to $X$ nowhere or everywhere,
but the latter can't occur by the symmetries of $R$ and~\eqref{posgenericity}.
Hence $V$ is spacelike and we normalise it.
Thus we can  choose an orthonormal Lipschitz frame $E_1,\dots ,E_n$ 
 on $U$ such that $E_1=V$, $E_n$ is timelike and $X=(E_{n-1}+E_n)$.

In the case where $\gamma$ is timelike, by shrinking $U$ and $c$ further, we may replace $V$ by
$V+\langle X,V\rangle X$ and normalize it. 
Consequently, there exists a Lipschitz continuous
orthonormal frame $E_1=V,E_2,\dots, E_n=X$ on $U$.

After these preparations, we will now carry out the proof in several steps simultaneously in the timelike and the null case.
 \medskip

 To begin with, let $0<c_1<c$. We claim that there exists some $C_1>0$ such
 that, setting $R_{ij}:= \langle R(E_i,X)X,E_j\rangle$ we have $(R_{ij})_{i,j=1}^{d} > \diag(c_1,-C_1,\dots,-C_1)$ on $U$.

 To establish this, we need to find $C_1>0$ such that, for any $w=:(w_1,\bar
 w)\not=0$ in $\R^{d}$,
 $w^\top (R_{ij}-\diag(c_1,-C_1,\dots,-C_1)) w >0$.
 Setting $\bar R := (R_{ij})_{i,j=2}^{d}$, and denoting by $\lamin$
 the smallest eigenvalue of $\bar R + C_1 \mathrm{id}$, we have
 \begin{equation}\label{rdiag}
  \begin{split}
   w^\top (R_{ij} -\diag(c_1,-C_1,\dots, & -C_1)) w \\
&= (R_{11}-c_1)w_1^2 + 2 \sum_{j=2}^{d} R_{1j} w_j w_1 +
   \bar w^\top (\bar R + C_1 \mathrm{id})\bar w \\
   & \ge (c-c_1)w_1^2 + 2 \sum_{j=2}^{d} R_{1j} w_j w_1 +
\lamin \|\bar w\|_e^2 \\
&\ge (c-c_1)w_1^2 - 2 |w_1| \|(R_{1j})_j\|_e \|\bar w\|_e + \lamin \|\bar
w\|_e^2,
\end{split}
\end{equation}
where $\|\,.\,\|_e$ denotes the Euclidean norm. Setting $C_R := \|(R_{1j})_j\|_e$, we pick $C_1 > 0$ such that $\lamin(x) \ge \frac{C_R^2}{c-c_1}$ for all $x \in U$. With this choice, the quadratic in the final line of~\eqref{rdiag} has no real root, and therefore~\eqref{rdiag} is positive for all $w \in \R^d \setminus \{ 0 \}$.
\medskip

 Since (component-wise) convolution with a non-negative mollifier as in Lemma~\ref{reg1}(ii) preserves
 positive-definiteness, it follows that given $0<c_2<c_1$ and
 $C_2>C_1$, we can achieve
 $R_{ij}*\rho_\eps > \diag(c_2,-C_2,\dots,-C_2)$ for $\eps$ small. Furthermore, by
 the same argument as in (5) in~\cite{hawkingc11},
 $R_\eps - R*\rho_\eps \to 0$ $(\eps\to 0)$ locally uniformly and, by Lemma
 ~\ref{reg1}(i),
 $R_{\eps ij}- R_{ij}*\rho_\eps \to 0$
 locally uniformly, where the matrix elements $R_{\eps
 ij}$ are defined as $R_{\eps ij}=(\langle R_\eps(E_i,X)X,E_j\rangle_{g_\eps})_{i,j=1}^{d}$.
 This implies that there exists an $\eps_0$ such that
 \begin{equation}\label{rdiag2}
  (R_{\eps ij}) > \diag(c_2,-C_2,\dots,-C_2)
 \end{equation}
 on $U$ for all $\eps<\eps_0$.

Next we note that by the explicit bounds derived
in~\cite[Sec.\ 2]{KSS} we may assume that
$U$ is $g_\eps$-totally normal for each $\eps<\eps_0$.
Let $p_\eps:=\gamma_\eps(0)$. Since $p_\eps\to p_0$,
we can also achieve that $p_\eps\in U$ for all $\eps<\eps_0$.
Pick a $g_\eps$-orthonormal
frame $e_1^\eps, \dots, e_n^\eps$ at $p_\eps$ such that,
as above, $e_n^\eps=\dot\gamma_\eps(0)$ in the timelike case,
whereas in the null case $e_n^\eps$ is timelike and
$\dot\gamma_\eps(0)\propto e_{n-1}^\eps+e_n^\eps$. In addition, we may
assume that $e_i^\eps\to E_i(p_0)$ as $\eps\to 0$. Now denote by
$E_1^\eps, \dots, E_n^\eps$ the $g_\eps$-orthonormal frame
on $U$ that results from parallel transporting $e_1^\eps, \dots,
e_n^\eps$ out from $p_\eps$ along radial $g_\eps$-geodesics.
Then, since $E_i^\eps\to E_i$ uniformly on $U$, by further
shrinking $\eps_0$, we
obtain from~\eqref{rdiag2} that the matrix elements with respect to this frame satisfy
\begin{equation}\label{rdiag3}
 (\langle R_\eps( E_i^\eps,X)X,
 E_j^\eps\rangle_{g_\eps})_{i,j=1}^{d}
  > \diag(c_2,-C_2,\dots,-C_2)
\end{equation}
on $ U$ for $\eps<\eps_0$.

Fix $r>0$ such that $\gamma([-r,r])\sse U$, so that, without loss of generality
we have $\gamma_\eps([-r,r])\sse U$ for all $\eps<\eps_0$.
Then, by construction, $E_i^\eps(t):= E_i^\eps\circ\gamma_{\eps}(t)$ is a
$g_\eps$-orthonormal smooth parallel frame along $\gamma_\eps$,
and~\eqref{rdiag3} implies that
\[
(\langle
R_\eps(E_i^\eps(t),X\circ\gamma_\eps(t))X\circ\gamma_\eps(t),
E_j^\eps(t)\rangle_ { g_\eps \circ \gamma_\eps})_{i,j=1}^{d} >
\diag(c_2,-C_2,\dots,-C_2)
\]
on $[-r,r]$ for $\eps\leq\eps_0$. The claim now follows from the observation that
$\langle R_\eps(\,.\,,X)X,\,.\,\rangle_{g_\eps} \circ \gamma_\eps -
\langle R_\eps(\,.\,,\dot \gamma_\eps)\dot \gamma_\eps,\,.\,\rangle_{g_\eps}
\to 0$
uniformly on $[-r,r]$.
\end{proof}
\section{Conjugate points for smooth metrics}
\label{sec:conjugate_points}
Given a causal geodesic $\gamma$ without conjugate points, it is well known in
the smooth case that, under the strong energy condition, the initial expansion
of the corresponding geodesic congruence
must be bounded. In the following Lemma, we explicitly derive such bounds
assuming only the weaker energy condition, $\Ric(\dot\gamma,\dot\gamma)>-\delta$, that follows from the
$C^{1,1}$-version of the strong energy condition, cf.\ Lemma~\ref{approxlemma}.
We respect the conventions introduced in Section~\ref{reg_sec}, so in particular
$d=n-1$ for $\gamma$ timelike and $d=n-2$ for $\gamma$ null.

\begin{Lemma}\label{boxing}
Let $g$ be a smooth Lorentzian metric on $M$. Then, for any $T>0$, there exists some $\delta=\delta(T)>0$
with the following property: Let $\gamma$ be a future
directed causal geodesic without conjugate points on $[-T,T]$, and let $[A]$
be the Jacobi tensor class along $\gamma$ assuming the data $[A](-T)=0$ and $[A](0)=\mathrm{id}$.
Then for any $0<r<T/2$
the expansion $\theta=\mathrm{tr}([\dot A] [A]^{-1})$ satisfies
\begin{align}\label{boxest}
 \sup_{t\in[-r,r]}|\theta(t)|\leq
 \frac{4d}{T},
\end{align}
provided that $\Ric(\dot\gamma,\dot\gamma)\geq-\delta$ on $[-T,T]$.
\end{Lemma}
\begin{proof} Since $[A](-T)=0$, $[B]:=[\dot A] [A]^{-1}$ is self-adjoint (cf., e.g., \cite[Lemma~4.6.19]{Krie}),
so its vorticity $\omega=\frac{1}{2}([B]-[B]^t)$ vanishes.
By the Raychaudhuri equation we therefore have
\begin{align}
 \dot\theta=-\Ric(\dot\gamma,\dot\gamma)-\mathrm{tr}(\sigma^2)
 -\frac{1}{d}\,\theta^2
 \leq \delta-\frac{1}{d}\,\theta^2,
\end{align}
where
the shear $\sigma$ is given by $\sigma=[B]-\frac{1}{d}\theta\cdot \mathrm{id}$.
To estimate $\theta$ from below on $[-r, r]$, assume that there exists $t_0 \in [-r,r]$ such that $\theta(t_0)<-\sqrt{d\delta}$. Writing
$\beta=\theta(t_0)<0$ and $\kappa=-\frac{1}{d}\delta<0$, we analyse the comparison equation
\begin{align}\label{box:comp}
 \dot s+\frac{1}{d}\,s^2+d\,\kappa =0,\qquad s(0)=\beta.
\end{align}
Denote by $s_{\kappa\beta} \colon [0,b_{\kappa\beta})\to\R$ the maximal solution of
$\eqref{box:comp}$. Now if $\beta\in(-\infty,-\sqrt{d\delta})$, one has (cf.\
\cite{TG})
\begin{align}
 s_{\kappa\beta}(t)&=d \sqrt{|\kappa|}\,
 \coth\left(t\,\sqrt{|\kappa|} + \mathrm{arcoth}\left(\frac{\beta}{d\sqrt{|\kappa|}}\right)\right) , \\
 b_{\kappa\beta}&=-\frac{1}{\sqrt{|\kappa|}}\,
  \mathrm{arcoth}\left(\frac{\beta}{d\sqrt{|\kappa|}}\right).
\end{align}
Since $\gamma$ has no conjugate point before $T$,
and since the maximal domain of definition of $\theta(t_0+\,.\,)$ must be contained
in that of $s_{\kappa\beta}$ by Riccati comparison, we obtain
$T-t_0\leq b_{\kappa\beta}$. Consequently,
\begin{equation}\label{box:lim1}
-d\sqrt{|\kappa|}\, \coth\left(\sqrt{|\kappa|}\,(T-r)\right)
\le -d\sqrt{|\kappa|}\,\coth\left(\sqrt{|\kappa|}\,(T-t_0)\right)\leq \beta.
\end{equation}
The left hand side of
\eqref{box:lim1} goes to $-d/(T-r)$ as $\kappa\to 0$,
so we may choose a $\kappa <0$ of small enough modulus such
that $\beta\geq-2d/(T-r)$.
Translating back to $\delta$ and recalling that we assumed $r\leq T/2$, we see that
we may choose $\delta>0$ small enough such that, for any $t_0$ as above,
$\beta=\theta(t_0)\geq -4d/T$.
So in total we have for sufficiently small $\delta$ that
\begin{align}\label{thetaest}
 \inf_{t\in[-r,r]}\theta(t)\geq
 \mathrm{min}(-\frac{4d}{T},-\sqrt{d\,\delta})=-\frac{4d}{T}.
\end{align}
To obtain the analogous estimate from above, consider the Jacobi tensor
$t\mapsto [A](-t)$ along $t\mapsto \gamma(-t)$. Then the corresponding past-directed
expansion $\theta_p(t)=-\theta(-t)$ satisfies a Riccati equation with the same bounds
as $\theta$, so the above arguments imply~\eqref{thetaest} also for $\theta_p$,
yielding the claim.
\end{proof}

We may now prove the existence of conjugate points along
causal geodesics in the smooth case under the weakened version of the Ricci bounds
derived in Lemma~\ref{approxlemma} from the strong energy condition
\eqref{eq:sec}, as well as the bounds on the curvature operator derived in Lemma
\ref{lem:2.3} from the $C^{1,1}$-genericity condition.

\begin{Proposition}\label{mel_prop}
Let $g$ be a smooth Lorentzian metric on $M$. Then given $c>0,\ C>0$, and $0<r<\frac{\pi}{4\sqrt{c}}$ there exist
$\delta=\delta(c,C,r)>0$, and $T=T(c,C,r)>0$ with the following property:

\noindent
If $\gamma$ is a causal geodesic and $t_0\in \R$ is
such that $\gamma$ is defined at least on $[t_0-T,t_0+T]$ and
\begin{itemize}
\item[(i)] $\Ric(\dot \gamma,\dot \gamma)\ge -\delta$ on $[t_0-T,t_0+T]$, as well as
\item[(ii)]
there exists a smooth parallel orthonormal frame $[E_1](t),\dots,$ $[E_{d}](t)$
for $[\dot{\gamma}]^\perp$ such that,
in terms of this frame the tidal force operator satisfies
$[R](t) > \diag(c,-C,\dots,-C)$ on $[t_0-r,t_0+r]$,
\end{itemize}
then $\gamma$ possesses a pair of conjugate points in $[t_0-T,t_0+T]$.
\end{Proposition}

\begin{proof}
Clearly we may assume that $t_0=0$. Now suppose, to the contrary, that no matter
how small $\delta>0$ or how big $T>0$ are chosen, there exists a $\gamma$
satisfying (i) and (ii) without conjugate points in $[-T,T]$. Then for any such
choice there is a unique Jacobi tensor class $[A]$ along $\gamma$ (depending on
$T$ and $\delta$) with $[A](-T)=0$ and $[A](0)=\mathrm{id}$. With
$[E_1](t),\dots,$ $[E_{d}](t)$ as in (ii), henceforth we will consider all
linear endomorphisms of $[\dot \gamma]^\perp$ as matrices in this basis. Set
$[\tilde R](t):=\diag(c,-C,\dots,-C)$. Then by (ii), $[\tilde R](t)<[R](t)$ on
$[-r,r]$.
\medskip

Set $[B]:=[\dot A]\cdot [A]^{-1}$. Then (cf., e.g.,~\cite[ch.\ 12]{BEE}) $[B]$
is self-adjoint and satisfies the matrix Riccati equation
\begin{equation}\label{riccati} [\dot B] + [B]^2 + [R] = 0. \end{equation}
Denote by $[\tilde B]$ the solution to~\eqref{riccati}, with $[R]$ replaced by
$[\tilde{R}]$ and initial value prescribed at some $t_1\in [-r,r]$.
We will show that we can find a $t_1 \in [-r, r]$ and an initial value $[\tilde B](t_1)$ satisfying $[\tilde B](t_1) \ge [B](t_1)$. Once this is established then, since $[R] > [\tilde R]$ on $[-r, r]$, the Riccati comparison theorem of~\cite{EH} implies that $[B](t) \le [\tilde B](t)$ for all $t \in [t_1, r]$.

We will in fact seek $t_1$ in $[-r,0]$ and $[\tilde B](t_1)$ in the form
$\tilde \beta(t_1)\cdot \mathrm{id}$, where $\tilde \beta(t_1)$ is greater or
equal the largest eigenvalue of $[B](t_1)$. Since we can without loss of
generality assume that $T>2r$ and that $\delta<\delta(T)$,
our assumption on the absence of conjugate points
in conjunction with Lemma~\ref{boxing} yields for the expansion
$\theta=\mathrm{tr}([B])$:
\begin{equation}\label{thetabound}
\max_{t\in [-r,r]}|\theta(t)| \le \nu \equiv \nu(T):=\frac{4d}{T}.
\end{equation}
Also, $\theta$ satisfies the Raychaudhuri equation
\begin{equation}
\dot \theta + \frac{1}{d}\theta^2 + \mathrm{tr}(\sigma^2) + \mathrm{tr}([R]) =
0,
\end{equation}
where, as before, $\sigma=[B]-\frac{1}{d}\theta\cdot
\mathrm{id}$. Denoting the eigenvalues of $[B]$ by $\beta_i$ ($1\le i\le d$),
$\sigma$ has eigenvalues $\beta_i - \frac{\theta}{d}$, and since
$\mathrm{tr}([R])\ge -\delta$ by assumption we find
\begin{equation}\label{thetadotest}
\dot \theta \le \delta - \sum_{i=1}^{d}\Big(\beta_i-\frac{\theta}{d}\Big)^2
\le \delta - \Big(\beta_{\max}(t_1)-\frac{\theta(t_1)}{d}\Big)^2
=:-l.
\end{equation}
Here, $\beta_{\max}$ is the maximum eigenvalue of $[B]$ and $t_1\in [-r,0]$ is
chosen such that $\left|\beta_{\max}-\frac{\theta}{d}\right|$ attains its
minimum on $[-r,0]$ in $t_1$. Using~\eqref{thetabound}, we see
\[
-\nu \le \theta(0) \le -l r +\theta(-r) \le -l r + \nu,
 \quad \text{which implies}\quad l \le \frac{2\nu}{r}. \]
Combining this with~\eqref{thetadotest} gives
\begin{equation}
\beta_{\max}(t_1) \le \sqrt{\left( \frac{2\nu}{r}+\delta \right)} + \frac{\theta(t_1)}{d}\le
\sqrt{\left( \frac{2\nu}{r}+\delta \right)} + \frac{\nu}{d} =: f(\nu,\delta,r)\equiv f.
\end{equation}
Consequently, we may set $\tilde \beta(t_1):=f(\nu,\delta,r)$ and $[\tilde
B](t_1):=f(\nu,\delta,r)\cdot\mathrm{id}$ to indeed achieve that $[B](t)\le
[\tilde B](t)$ on $[t_1,r]$.

Since both $[\tilde R]$ and $[\tilde B](t_1)$ are diagonal, the Riccati
equation for $[\tilde B]$ decouples and has the
explicit solution
\[
[\tilde B](t) = \frac{1}{d}\diag(H_{c,f}(t),H_{-C,f}(t),\dots,H_{-C,f}(t)).
\]
Here (cf.~\cite{TG,G})
\[
H_{c,f}(t) = d \sqrt{c}\cot(\sqrt{c}(t-t_1) + \mathrm{arccot}(f/\sqrt{c})),
\]
and
\[
H_{-C,f}(t)=d \sqrt{C}\tanh\big(\sqrt{C}(t-t_1)+\mathrm{artanh}(f/\sqrt{C})\big),
\]
and due to our assumption $0<r<\frac{\pi}{4\sqrt{c}}$ these functions are defined on $[t_1,r]$
(for $f$ sufficiently small).
As was noted above, since $[\tilde R](t)<[R](t)$ for all $t\in [-r,r]$ and
$[B](t_1)\le [\tilde B](t_1)$, Riccati comparison implies
$[B](t)\le [\tilde B](t)$ for all $t\in [t_1,r]$. In particular, for the
smallest eigenvalue $\beta_{\min}$ of $[B]$ we obtain
\begin{equation}\label{betaminest}
\beta_{\min}(t) \le \frac{1}{d}H_{c,f}(t) \qquad (t\in [t_1,r]).
\end{equation}
We are now going to show that for $\delta$ small enough and $T$ large enough, $H_{c,f}(t)<0$ for $t\in [\frac{r}{2},r]$.
In fact, since $H_{c,f}$ is monotonically decreasing, it suffices to secure that $H_{c,f}(\frac{r}{2})<0$.
Set $k:=\mathrm{arccot}(f/\sqrt{c})<\frac{\pi}{2}$. Then $H_{c,f}(\frac{r}{2})<0$ if and only if
\begin{equation}\label{kequ}
\sqrt{c}\Big(\frac{r}{2}-t_1\Big) + k \in \Big(\frac{\pi}{2},\pi\Big).
\end{equation}
To achieve this, first note that $3r\sqrt{c}<\pi$, so that
$\sqrt{c}(\frac{r}{2}-t_1)<\frac{\pi}{2}$. Since $k<\frac{\pi}{2}$,~\eqref{kequ}
can be satisfied by choosing $\delta$ and $\nu$ so small that
$\sqrt{c}(\frac{r}{2}-t_1) + k >\frac{\pi}{2}$. Shrinking $\nu$ further, we can
also achieve that $H_{c,f}(\frac{r}{2})<-\nu$, so altogether we obtain for $t\in
[\frac{r}{2},r]$:
\[
\beta_{\min}(t) \le \frac{1}{d} H_{c,f}\Big(\frac{r}{2}\Big) <
-\frac{\nu}{d}\le \frac{\theta(t)}{d}.
\]
By~\eqref{thetadotest} this gives
\[
\dot{\theta} \le -\Big(\beta_{\min}-\frac{\theta}{d}\Big)^2 + \delta
\le - \Big(\frac{1}{d}\Big(H_{c,f}\Big(\frac{r}{2}\Big)+\nu\Big)\Big)^2 +
\delta
\]
on $[\frac{r}{2},r]$. Consequently,
\[
-2\nu \le \int_{\frac{r}{2}}^{r}\dot\theta(t)\,dt\le
-\frac{r}{2}\Big[\Big(\frac{1}{d}\Big(H_{c,f} \Big(\frac{r}{2}\Big) +
\nu\Big)\Big)^2-\delta\Big],
\]
and thereby
\[
-d\sqrt{\left( \frac{4\nu}{r}+\delta \right)} - \nu \le H_{c,f}\Big(\frac{r}{2}\Big).
\]
However, as $\delta\searrow 0$ and $T\to \infty$, the left hand side of this inequality tends to $0$,
while the right hand side has the limit $d\sqrt{c}\cot\Big(\sqrt{c}\Big(\frac{r}{2}-t_1\Big)+\frac{\pi}{2}\Big)<0$,
a contradiction.
\end{proof}
\section{Maximising geodesics}
\label{sec:max_geods}
We will next prove that in the $C^{1,1}$-case under suitable causality
conditions complete causal geodesics stop being maximising, provided the strong
energy condition~\eqref{eq:sec} and the genericity condition (Definition~\ref{genericity_c11}) hold. We will do so separately in
the timelike and in the null case with the respective causality conditions
adapted to the later use of the corresponding statements in the proof of the
main theorem.

\begin{Theorem}\label{timelikenotmax}
Let $g\in C^{1,1}$ be a globally hyperbolic Lorentzian metric on $M$ that
satisfies the timelike convergence condition. Moreover, suppose that the genericity condition holds
along any timelike geodesic. Then no complete timelike geodesic $\gamma \colon \R\to M$ is globally maximising.%
\footnote{Recall that a timelike geodesic is~\emph{globally maximising\/} if it maximises between any two of its points.}
\end{Theorem}
\begin{proof}
Let $\gamma \colon \R\to M$ be a complete geodesic and suppose that $\gamma$ were
maximising between any two of its points. We approximate $g$ from the inside by
a net $\check{g}_\eps$, so each $\check g_\eps$ is globally hyperbolic as well.
Without loss of generality assume that $\gamma$ satisfies the genericity
condition at $t_0=0$. Then by Lemma~\ref{lem:2.3} there exist $c>0$, $C>0$ and $0<r<\frac{\pi}{4\sqrt{c}}$
such that, whenever $\gamma_\eps$ is a net of $\check g_\eps$-geodesics that
converge to $\gamma$ in $C^1$, there exists some $\eps_0>0$ such that, for any
$\eps<\eps_0$, condition (ii) of Proposition~\ref{mel_prop} is satisfied for
$R_\eps$.

Choose $\delta=\delta(c,C,r)>0$ and $T=T(c,C,r)>0$ as in Proposition~\ref{mel_prop} and let $\tilde T>T$. Since $\check g_\eps$ is globally hyperbolic,
for any $\eps>0$ sufficiently small there exists a maximising $\check g_\eps$-geodesic $\gamma_\eps$ from $\gamma(-\tilde T)$ to $\gamma(\tilde T)$ (cf.~\cite[Prop.\ 1.21 and Th.\ 1.20]{CG}). We choose the parametrisation such that $\gamma_{\eps}(-\tilde{T})=\gamma(-\tilde{T})$ and $v:=\dot \gamma(-\tilde T)$ and $v_\eps := \dot\gamma_\eps(-\tilde T)$ have the same $h$-norm for a fixed Riemannian background metric $h$. We define $\tilde{T}_\eps $ by $\gamma_{\eps}(\tilde{T}_\eps)=\gamma(\tilde{T})$, so $\gamma_{\eps}|_{[-\tilde{T},\tilde{T}_\eps ]}\sse J^-(\gamma(\tilde{T}))\cap J^+(\gamma(-\tilde{T}))$. Therefore there is a subsequence $\eps_k$ such that $v_{\eps_k}$ converges to a vector $w$ with $\|w\|_h=\|v\|_h$ and $\tilde{T}_{\eps_k} \to b \in [-\tilde{T},\infty]$.

Consequently, $\gamma_{\eps_k}$
converges in $C^1$ to the (future) inextendible $g$-geodesic $\gamma_w \colon [-\tilde{T},b_0)\to M$ with $\gamma_w(-\tilde T)=\gamma(-\tilde T)$ and $\dot \gamma_w(-\tilde T)=w$. Since our spacetime is non-totally imprisoning (which follows from global hyperbolicity by the same proof as for smooth metrics,~\cite[Lem.~14.13]{ON83}), this geodesic must leave the compact set $J^-(\gamma(\tilde{T}))\cap J^+(\gamma(-\tilde{T}))$, hence $b_0>b$ and in particular $b\neq \infty$ and $\gamma_w (b)=\gamma(\tilde{T})$. Also, $\gamma_w|_{[0,b]}$ must be maximising since the distances converge by Lemma~\ref{d_conv}.
We now distinguish two cases:

If $w\not=v$, then $\gamma_w$ is a maximising geodesic from $\gamma(-\tilde T)$ to $\gamma(\tilde T)$ different from $\gamma$,
so $\gamma$ can't be maximising beyond $\tilde T$, contradicting our assumption.

If, on the other hand, $v=w$, then $\gamma_w=\gamma$ and $b=\tilde{T}$. Let $K$ be a compact neighbourhood of $\gamma([-T,T])$.
Since $\gamma_{\eps_k}\to \gamma$ in $C^1([-T,T])$, there exist $k_0\in \N$, $\tilde C>0$, and $\kappa<0$ such that
for all $k\ge k_0$ we have $\gamma_{\eps_k}([-T,T])\sse K$, as well as
$\|\dot \gamma_{\eps_k}(t)\|_h\le \tilde C$ and $g(\dot \gamma_{\eps_k}(t),\dot \gamma_{\eps_k}(t))<\kappa$ for all
$t\in [-T,T]$. Lemma~\ref{approxlemma}(i) therefore implies that
$R_{\eps_k}(\dot \gamma_{\eps_k}(t),\dot \gamma_{\eps_k}(t))
\ge -\delta(c,C,r)$ on $[-T, T]$ for $k$ sufficiently large. This shows that $\gamma_{\eps_k}$ also satisfies condition
(i) from Proposition~\ref{mel_prop} for $k$ large. But then any such $\gamma_{\eps_k}$ incurs a pair of conjugate points
within $[-T,T]$, contradicting the fact that it was supposed to be maximising even on $[-\tilde{T},\tilde{T}_{\eps_k}]\supset [-T,T]$ since $\tilde{T}_{\eps_k}\to \tilde{T}$.
\end{proof}

	The proof of the previous Theorem uses Proposition~\ref{mel_prop} to guarantee the existence of conjugate points for $\check{g}_\eps$-geodesics close to $\gamma $, but the essence of the argument can be formulated in a much more general way using cut functions. Let $\mathcal{T} \sse TM$ be the set of all future directed timelike vectors, then one defines the timelike cut function $s \colon \mathcal{T}\to \R$ by
	\begin{equation}
s(v):=\sup \{t:\,L(\gamma_v|_{\left[0,t\right]})=d(\gamma(0),\gamma(t)) \}	.
	\end{equation}
This function clearly depends on the metric and so a natural question is how, given a $C^{1,1}$-metric $g$, the $\check{g}_{\eps_k}$-cut functions $s_k$ relate to the $g$-cut function $s$. The following theorem shows that at least for a globally hyperbolic spacetime a uniform upper bound on the $s_k$ must also be an upper bound for $s$.
\begin{Theorem}\label{cutfunctionapprox}
	Let $(M,g)$ be a spacetime with a globally hyperbolic $C^{1,1}$-metric and let $g_k=\check{g}_{\eps_k}$.
	Let $U\sse \mathcal{T}$ be open such that $U\sse\mathcal{T}_k$ for large $k$. If $s_k|_U \leq T$ then $s|_U\leq T$.
\end{Theorem}
\begin{proof}
	The proof uses the same arguments as in Theorem~\ref{timelikenotmax}: Let $v\in U$, $\tilde{T}>T$ and assume, for the sake of contradiction, that $s(v)>\tilde{T}$. Then $\gamma_v$ maximises the distance between $\gamma_v(0)$ and $\gamma_v(\tilde{T})$ and even remains maximising a bit further. Choosing $\gamma_k$ as in the previous proof, the same arguments give a sequence $\gamma_k $ that converges in $C^1$ to $\gamma $ (in particular, $\dot{\gamma}_k(0)\in U$ for large $k$) and is maximising on $[0,\tilde{T}_k]\supset [0,T]$ for large $k$, but this contradicts $s_k|_U\leq T$.
\end{proof}

There is an analogous result to Theorem~\ref{timelikenotmax} for null instead of timelike curves. However, assuming global hyperbolicity in the null case renders such a statement mostly useless for the proof of the Hawking--Penrose Theorem because inextendible yet maximising null curves need to be excluded everywhere in the spacetime and not just in some globally hyperbolic subset (contrary to timelike curves, which will appear only briefly at the end of the proof when one already works in some Cauchy development). Fortunately in the null case there is a sharper distinction between maximising and non-maximising geodesics because a null geodesic stops maximising if and only if it leaves the boundary of a lightcone, and one can exploit the structure of such boundaries to show that inextendible null geodesics which are not closed cannot be maximizing. However, the methods of the following proof fail for closed null curves (which are not well behaved with respect to approximation), so these had to be excluded in the statement of Theorem~\ref{HPinC1,1} by assuming that the spacetime is causal instead of merely chronological in the classical theorem.

\begin{Theorem}\label{3.4}
 Let $g\in C^{1,1}$ be a Lorentzian metric on $M$ such that
 $(M,g)$ is causal.  Moreover, suppose that the null convergence
 condition holds and that the genericity condition is satisfied
 along any null geodesic. Then no complete null geodesic $\gamma \colon \R\to M$ is
 globally maximising.
\end{Theorem}

\begin{proof}
 The general shape of the argument is similar to the timelike case, however,
 since we do not assume global hyperbolicity we will have to choose the
 approximating $\check{g}_\eps$-geodesics differently.

 Assume $\gamma \colon \R \to M$ were a null geodesic that is maximizing between any
 of its points and that without loss of generality satisfies the genericity condition at
 $t_0=0$. Then by Lemma~\ref{lem:2.3} there
 exist $c>0$, $C>0$ and $0<r<\frac{\pi}{4\sqrt{c}}$ such that, whenever $\gamma_\eps$ is a net of $\check
 g_\eps$-null geodesics that
 converge to $\gamma$ in $C^1$, there exists some $\eps_0>0$ such that, for
 any $\eps<\eps_0$, condition (ii) of Proposition~\ref{mel_prop}
 is satisfied for $R_\eps$. Choose $\delta=\delta(c,C,r)>0$ and $T=T(c,C,r)>0$ as
 in Proposition~\ref{mel_prop} and choose $\tilde{T}>T$ in a such a way that $p:=\gamma(-\tilde{T})$ is different from $q:=\gamma(\tilde{T})$.

Then, by assumption,
 $q\in \partial J^+(p)$. We will now find a sequence $\eps_k \to 0$ and
 points $q_k \in \partial J^+_k(p):=\partial J^+_{\check g_{\eps_k}}(p)$ with
 $q_k \to q$: Let $U_k$ be a sequence of neighbourhoods of $q$ with $U_{k+1} \sse U_k$ and $\bigcap_k U_k =\{q\}$. Then for any $U_k$ there exist points $q_k^e \in U_k\setminus \overline{J^+(p)} $ and $q_k^i \in U_k \cap I^+(p)$. Let $\eps_k $ be such that $q_k^i \in I^+_{k}(p)$ and $\eps_k \leq \frac{1}{k}$ and let $c_k $ be a curve in $U_k$ connecting $q_k^i$ and $q_k^e\in U_k\setminus \overline{J^+(p)} \sse U_k\setminus \overline{J^+_{k}(p)}$. Then this curve must intersect $\partial J^+_{k}(p)$ and we choose $q_k$ to be such an intersection point.

 Since $q_k \in \partial J^+_k (p)$ there exists a past directed
 $\check g_{\eps_k}$-null geodesic starting at $q_k$ that is contained in
 $\partial J^+_k (p)$ and is either (past) inextendible or ends in $p$ (cf.\ Proposition
 \ref{prop:4}). Let
 $\gamma_k \colon I_k \to M$ denote an inextendible future directed reparametrisation of
 such a geodesic with $\gamma_k(\tilde{T})=q_k$ and
 $\|\dot{\gamma }_k(\tilde{T})\|_h=\|\dot{\gamma}(\tilde{T})\|_h$. Since the
 $h$-norms of $\dot{\gamma}_k(\tilde{T})\in T_{q_k}M$ are bounded and $q_k\to q$, we may without loss of generality
  assume that the sequence $\dot{\gamma}_k(\tilde{T})$ converges to some
vector $w\in T_qM$. This vector $w$ must be $g$-null since the $\dot{\gamma}_k$
were $\check{g}_{\eps_k}$-null.  Hence there exists
a unique inextendible $g$-geodesic $\gamma_w \colon (a_w,b_w)\to M$ with $\tilde{T}\in (a_w,b_w)$,
$\gamma_w(\tilde{T})=q$ and $\dot{\gamma}_w(\tilde{T})=w$ and the $\gamma_k$ converge to $\gamma_w$ in $C^1$.

Due to our choice of the $\gamma_k$, for each $k$ there either exists $t_k<\tilde{T}$ such
that $\gamma_k(t_k)=p$ and $\gamma_k|_{[t_k,\tilde{T}]}\sse \partial J^+_k(p)$ or
$\gamma_k \sse \partial J^+_k(p)$. By extracting a subsequence we may assume that the first or
the second possibility applies in fact for each $k$.
In the second case we pick some $s\in (a_w,\tilde T)$
and note that by $C^1$-convergence $\gamma_k$ is defined on $[s,\tilde T]$ for $k$ large.

In the first case, if the sequence $t_k$ is unbounded (below) we may again pick
some $s\in (a_w,\tilde T)$ such that $\gamma_k([s,\tilde T])\sse \partial J_k^+(p)$ for $k$ large.
Finally, if $(t_k)$ is bounded,
we may without loss of generality assume that $t_k\to \tilde{t}$
with $\gamma_w (\tilde{t})=p$.
Since $p\neq q$ (by our choice of $\tilde{T}$), $\tilde{t}<\tilde{T}$, so also in this case there exists
$\max(\tilde{t},a_w)<s<\tilde{T}$ such that $\gamma_k([s,\tilde{T}])\sse \partial J^+_k(p)
\sse \overline{J^+(p)}$ for large $k$.

Thus in any case $\gamma_w|_{[s,\tilde{T}]}\sse \overline{J^+(p)}$. Therefore,
if $\gamma_w$ were not (a reparametrisation of) $\gamma $, the concatenation
$\gamma_w|_{[s,\tilde{T}]} \gamma|_{[\tilde{T},\tilde{T}+1]}$ would be a
broken null curve from a point in $\overline{J^+(p)}$ to $\gamma (\tilde{T}+1)$,
hence $\gamma (\tilde{T}+1)\in
I^+(p)$, which contradicts $\gamma $ being maximising between any of its points.
This shows that (with our choice of parametrisations) $\gamma_w$ must
actually be equal to $\gamma$.

But then in particular $\gamma (\tilde{t})=\gamma_w(\tilde{t})=p$ (if $(t_k)$ is bounded) and thus since $\gamma $ cannot be closed by assumption of causality, we must have $t_0=-\tilde{T}$. Thereby in each of the above cases $\gamma_k|_{[-T,\tilde{T}]} \sse \partial J^+_k(p)$ for $k$ large. Consequently, any
such segment must be maximising for the metric $\check{g}_{\eps_k}$.
Also, since $\gamma_k \to \gamma $ in $C^1([-T,T])$, there
exist a compact neighbourhood $K$ of $\gamma ([-T,T])$, $k_0\in \N $, $\tilde C>0$, and $\eta >0$ such that
for all $k\ge k_0$ we have $\gamma_k([-T,T])\sse K$, as well as
$\|\dot \gamma_k(t)\|_h\le \tilde C$ and $d_{\tilde{h}}(\dot \gamma_k(t),\dot \gamma(t))<\eta $
and $\|\dot \gamma(t)\|_h\le \tilde C$ for all
$t\in [-T,T]$. Lemma~\ref{approxlemma}(ii) therefore implies that
$R_{\eps_k}(\dot
\gamma_k(t),\dot \gamma_k(t))
\ge -\delta(c,C,r)$ on $[-T, T]$ for $k$ sufficiently large. This shows that $\gamma_k$
also satisfies condition
(i) from Proposition~\ref{mel_prop} for $k$ large. But then any such $\gamma_k$ incurs a
pair of conjugate points
within $[-T,T]$, contradicting the fact that it was supposed to be maximising
even on $[-T,\tilde T]$.\end{proof}

To conclude this section we want to briefly discuss the difference in causality conditions imposed on $M$ in the classical Theorem~\ref{classicalHP} ($M$ being chronological) and in the $C^{1,1}$-Theorems~\ref{HPinC1,1} and~\ref{ArbitrCodiminC1,1} ($M$ being causal). Causality assumptions (of any kind) on $M$ were first required in this section to prove Theorem~\ref{timelikenotmax} and Theorem~\ref{3.4}. The results proven in previous sections did not require any causality assumption (with the exception of Lemma~\ref{d_conv}, which is only used in the proof of Theorem~\ref{timelikenotmax}). Contrary to our results the smooth versions of these two theorems do not require any causality conditions. Regarding Theorem~\ref{timelikenotmax}, we note that even in the proof of the (classical) Hawking--Penrose theorem its smooth counterpart (despite being valid on all of $M$) is actually only applied to an open globally hyperbolic subset of $M$. This is also true in the proof of our result (see Theorem~\ref{C11HPCausalitybit}). However, Theorem~\ref{3.4} is required in multiple places (e.g., any result requiring strong causality indirectly uses Theorem~\ref{3.4} by virtue of Lemma \ref{lem: 19 conj. pts imply strong causality}). As such, we have found it necessary to assume that the $C^{1, 1}$-spacetime is causal.

Nevertheless, the assumption of causality of $M$ only enters in the proof of Theorem~\ref{3.4} at a single point, namely where we argue that since $\gamma $ cannot be closed the equality of $\gamma (\tilde{t})$ and $\gamma(-\tilde{T})$ implies that $\tilde{t}=-\tilde{T}$. Moreover, this theorem is the only ingredient in the proof of Theorems \ref{HPinC1,1} and \ref{ArbitrCodiminC1,1} where causality of $M$ is required. For all other steps it is sufficient that $M$ be chronological. This can be seen from the following argument: Both the classical proof of the Hawking--Penrose theorem and the proofs of Theorem~\ref{HPinC1,1} and Theorem~\ref{ArbitrCodiminC1,1}  presented here argue by contradiction, i.e., one assumes that $M$ is a causal geodesically complete spacetime (satisfying the  conditions of the theorem) and derives a contradiction. Hence if one could show that Theorem \ref{3.4} remains true while only assuming $M$ to be chronological (and not causal), one could invoke Lemma \ref{lem: 19 conj. pts imply strong causality} to gain that $M$ is even strongly causal and the rest of our proof would go through.

We expect that Theorem~\ref{HPinC1,1} and Theorem~\ref{ArbitrCodiminC1,1}, in fact, even hold for chronological $C^{1, 1}$ spacetimes, but anticipate that a proof will require new methods. 

\section{Initial conditions}
\label{sec:trapped}
In its classical version the Hawking--Penrose theorem comes with three distinct initial conditions: the existence of a compact achronal set without edge (or equivalently an achronal compact topological hypersurface,~\cite[Cor.\ A.19]{hawkingc11}), the existence of a trapped surface, or the existence of a point such that along any future (or past) directed null geodesic starting at this point the convergence becomes negative. An analogue of the trapped surface condition for submanifolds of arbitrary co-dimension was introduced in~\cite{GS}. In this section we will study these initial conditions and their consequences in the $C^{1,1}$-case.

\subsection{The hypersurface case}
We begin with the most straightforward case: the existence of a compact achronal set without edge.

\begin{Proposition}
\label{hypersurfacecase} Let $(M,g)$ be a $C^{1,1}$-spacetime, and
let $A$ be a compact achronal set without edge. Then $E^+(A)=A$, in particular it is compact.
\end{Proposition}
\begin{proof} This follows immediately from the fact that for an achronal set $A$ any future directed null geodesic starting in a point $p\notin \mathrm{edge}(A)$ must immediately enter $I^+(A)$. This can be seen as in~\cite[Prop.\ A.18]{hawkingc11}.
\end{proof}
One should note that as in the smooth case one may even relax the causality assumptions on $A$ a little: By using a covering argument as in~\cite[Thm.\ A.34]{hawkingc11} it would be sufficient to assume the existence of a compact spacelike hypersurface $A$ in the Hawking--Penrose theorem.

\subsection{Submanifolds of codimension $1<m<n$}

In this section, we follow the approach of Galloway and Senovilla~\cite{GS} and
consider trapped submanifolds
of arbitrary codimension of a $C^{1,1}$-spacetime $(M,g)$.
To work in full generality (and because we will need this generality to deal
with the codimension zero case later on) we will now define $C^0$-trapped
submanifolds of codimension $1<m<n$. Our definition is similar in spirit to the
definition of lower mean curvature bounds for $C^0$ spacelike hypersurfaces in
\cite{AGH}.

As mentioned in section~\ref{sec:mainresult},
we say that a submanifold $\tilde{S}$
is a future support submanifold for a $C^0$-submanifold $S$ at $q\in S$ if $\dim
(\tilde{S})=\dim S$, $q\in \tilde{S}$, and $\tilde{S}$ is locally to the future of $S$
near $q$, i.e. there exists a neighbourhood $U$ of $q$ in $M$ such that
$\tilde{S}\cap U \sse J^+(S,U)$. We use this to define 'past pointing timelike mean curvature' for $C^0$-submanifolds.

\begin{definition}\label{support_submf}
 Let $S$ be a $C^0$-submanifold of codimension $m$ ($1<m<n$) in a
$C^{1,1}$-spacetime $(M,g)$. We say that $S$
  has past-pointing timelike mean curvature in $q$ in the sense of
 support submanifolds if there exists a $C^2$
 spacelike future support submanifold $\tilde{S}$ for $S$ in $q$ with $H_{\tilde{S}}(q)$ past-pointing timelike.
\end{definition}

This leads to the following definition of a future trapped $C^0$-submanifold of $M$
(which is obviously satisfied for $C^2$-submanifolds that are future trapped in the classical sense defined in~\cite{GS}).

\begin{definition}
 A $C^0$-submanifold $S$ of codimension $m$ ($1<m<n$) of a
$C^{1,1}$-spacetime $(M,g)$ is called
 future trapped if it is closed (i.e., compact without boundary) and
for any $p\in S$ there exists a neighbourhood $U_p$ of $p$ such that $S\cap U_p$
is achronal in $U_p$ and $S$ has past-pointing timelike mean
 curvature in all its points (in the sense of support submanifolds).
\end{definition}

Our aim is a generalisation of the main results of~\cite{GS} to the
$C^{1,1}$-setting. In fact, we will
show that under some additional curvature assumptions any future directed null
geodesic starting at a
point $q$ of a trapped submanifold $S$ in the above sense eventually stops
maximising the distance to the
future support submanifold $\tilde{S}$ at $q$ (provided it exists for long enough
times).

Using the notation introduced in section~\ref{sec:intro} (i.e., letting $E_1,\dots,E_{n-m}$ denote the parallel translates of an orthonormal basis $e_1(\gamma(0)),\dots,e_{n-m}(\gamma(0))$ for $T_{\gamma(0)}S$ along $\gamma$)
we start by proving the following mild extension of
\cite[Prop.\ 1]{GS}:

\begin{Lemma} \label{GS_refine}
 Let $S$ be a $C^2$ spacelike submanifold of codimension $m$ $(1<m<n)$ in a smooth spacetime $(M,g)$, and let $\gamma $ be a geodesic such that
 $\nu := \dot{\gamma}(0)\in TM|_{S}$ is a
 future-pointing null normal to $S$.
 Suppose that $c := \mathbf{k}_S(\nu)>0$ and let $b>1/c$. Then there exists some
  $\delta=\delta(b,c)>0$ such that, if
  \begin{equation}
  \sum_{i=1}^{n-m} \langle R(E_i,\dot{\gamma})\dot{\gamma},E_i \rangle \geq -\delta
  \label{RicciArbCodim_smooth}
 \end{equation}
 along $\gamma$, then $\gamma|_{[0,b]}$ is not maximising to $S$, provided
 that $\gamma$ exists up to
 $t=b$.
\end{Lemma}
\begin{proof}
 We closely follow the proof of~\cite[Prop.\ 1]{GS}. For vector fields $V$,
 $W$ along $\gamma$ that are orthogonal
 to $\gamma$ and vanish at $t=b$ we consider the energy index form (with $\dot{V}$
 etc.\ denoting the induced covariant  derivative along $\gamma$)
 \[
 I(V,W):=\int_0^b \left[ \langle \dot{V},\dot{W}\rangle - \langle R_{V\dot \gamma}\dot
 \gamma,W\rangle  \right] \,dt
 - \langle \dot \gamma(0),\mathrm{II}(V(0),W(0))\rangle.
 \]
 For $1\le i\le n-m$, let $X_i:=(1-t/b)E_i$. Then
 \[
 I(X_i,X_i)=\int_0^{b} \left[ \vphantom{|^|} 1/b^2 - (1-t/b)^2 \langle R_{E_i\dot \gamma}\dot
 \gamma,E_i\rangle \right] \,dt
 - \langle \dot \gamma(0),\mathrm{II}(e_i,e_i)\rangle.
 \]
 Hence
 \begin{equation*}
  \begin{split}
   \sum_{i=1}^{n-m}I(X_i,X_i) &= (n-m)\Big(\frac{1}{b}-c\Big) -\int_0^{b}
   \Big(1-\frac{t}{b}\Big)^2 \sum_{i=1}^{n-m} \langle R(E_i,\dot{\gamma})\dot{\gamma},E_i \rangle\,dt\\
   & \le (n-m)\Big(\frac{1}{b}-c\Big) +
   \frac{b\delta}{3}.
  \end{split}
 \end{equation*}
 Obviously this last expression can be made negative by choosing
 $\delta=\delta(b,c)$ small enough. It then follows that
 the energy index form is not positive-semidefinite, so there must exist a
 focal point of $S$ on $\gamma$ within
 $(0,b]$, giving the claim.
\end{proof}
We now turn to the case of a $C^{1,1}$-metric $g$. Let $\tilde{S}$ be a $C^2$
spacelike submanifold of co-dimension $m$,
and let $\nu\in T_p\tilde S$ be a future-pointing
null vector normal to $\tilde S$.
As in the smooth setting above, assume that
$\gamma$ is a geodesic with affine parameter $t$ with $\dot
\gamma(0)=\nu $, and let $e_1,\dots,e_{n-m}$
be a local orthonormal frame on $\tilde S$ around $p:=\gamma(0)$ (of regularity
$C^{1,1}$). Again, denote by $E_1,\dots,E_{n-m}$ the parallel
translates of $e_1(p),\dots,e_{n-m}(p)$ along $\gamma$ (which are Lipschitz continuous
vector fields along $\gamma$).

In trying to formulate a natural analogue of~\eqref{RicciArbCodim_smooth} (with $\delta=0$) we
again face the problem
that the curvature operator (being only defined almost everywhere) cannot be
restricted to the Lebesgue null set $\gamma([0,b])$. Similar to the case of
the genericity condition (Definition~\ref{genericity_c11}),
we shall therefore require the existence of continuous extensions
of $E_1,\dots E_{n-m}$ and $\dot{\gamma}$ to a neighbourhood of the geodesic $\gamma$. In
fact, with the notation
introduced above we have:
\begin{Proposition} \label{notmaxtosupportsubmf} Let $(M,g)$ be a strongly causal $C^{1,1}$-spacetime, $\tilde{S}\sse M$ a $C^2$ spacelike submanifold and
 suppose that $\mathbf{k}_{\tilde{S}}(\nu)>c>0$ and let $b>1/c$.
 If there exists a neighbourhood $U$ of $\gamma|_{[0,b]} $ and continuous
 extensions $\bar{E}
 _1,\dots \bar{E}_{n-m}$ and $\bar{N}$ of $E_1,\dots E_{n-m}$ and $\dot{\gamma}$,
 respectively, to $U$ such that
 \begin{equation}
  \sum_{i=1}^{n-m} \langle R(\bar{E}_i,\bar{N})\bar{N},\bar{E}_i \rangle \geq
  0,
  \label{RicciArbCodim}
 \end{equation}
 then $\gamma|_{[0,b]}$ is not maximising to $\tilde{S}$.
\end{Proposition}

\begin{proof} We again proceed by regularisation.
Let $g_\eps = \check g_\eps$, then as in the proof of Lemma~\ref{lem:2.3} we
may
without loss of
generality suppose that $M=\R^n$, and that $R_\eps=R*\rho_\eps$.
Since $\tilde{S}$ is a $C^2$-submanifold, $\mathbf{k}_{\tilde{S}}$ is continuous on $\tilde{S}$ and $\mathbf{k}_{\tilde{S}, \eps} \to
\mathbf{k}_{\tilde{S}}$ uniformly on compact subsets. Thus, there exists a neighbourhood
$V$
in $TM|_{\tilde{S}}$ of $\nu$ and an $\varepsilon_0 $ such that
for all $\varepsilon \leq \varepsilon_0 $ one has $\mathbf{k}_{\tilde{S}, \varepsilon}
(v)>c$
for all $v\in V$.
Shrinking $U$, we may assume that there exists $\eps_0$ such that for
all $g_\eps$ with $\eps\le\eps_0$ the submanifold $U\cap \tilde S$ is $g_\eps$-spacelike and, shrinking $V$ if necessary, we have that the projection $W:=\pi(V)$ of $V$ onto $\tilde{S}$ is contained in $U\cap \tilde S$.

Further shrinking $\eps_0$ and $V$ if necessary, for each $\eps<\eps_0$ let
$e_1^\eps,\dots,e_{n-m}^\eps$ be a $g_\eps$-orthonormal frame for $\tilde S$ on $W$
such that $e_i^\eps \to e_i$ uniformly on $W$ for
$\eps\to 0$. For each $v\in V$, denote by $E_i^\eps(t)$ the parallel transport
of $e_i^\eps(\pi(v))$ along the $g_\eps$-geodesic $\gamma_v^\eps$ with
$\dot\gamma_v^\eps(0)= v$.

By~\eqref{RicciArbCodim} we have \[ \sum_{i=1}^{n-m}
g(R(\bar{E}_i,\bar{N})\bar{N},\bar{E}_i)*\rho_\eps \geq 0. \] Since without
loss of generality $U$ is relatively compact and $\gamma^\eps_v([0,b])\sse U$ for
all $v\in V$ and all $\eps\le \eps_0$, Lemma~\ref{reg1} (i) implies that
$g(R(\bar{E}_i,\bar{N})\bar{N},\bar{E}_i)*\rho_\eps -
g_\eps(R_\eps(\bar{E}_i,\bar{N})\bar{N},\bar{E}_i)\to 0$ uniformly on $U$, as
well as \[ g_\eps(R_\eps(\bar{E}_i,\bar{N})\bar{N},\bar{E}_i)\circ\gamma_v^\eps
- g_\eps(R_\eps(E_i^\eps,\dot \gamma_v^\eps)\dot\gamma_v^\eps,E_i^\eps) \to 0
\]
uniformly on $[0,b]$ as $(\eps,v)\to (0,\nu)$, for $1\le i\le n-m$.

Now let $1/c<b'<b$, and pick $\delta:=\delta(b',c)$ as in Lemma
\ref{GS_refine}.
Then by the above we may shrink $V$ and $\eps_0$ in such a way that
condition~\eqref{RicciArbCodim_smooth} is satisfied along $\gamma_v^\eps$ on
$[0,b']$
for each $v\in V$ and each $\eps\le \eps_0$.

Consequently, any $\gamma^\eps_v$ with $v$ being $g_\eps$-null stops maximising the
$g_\eps$-distance to $\tilde S$ at parameter $t=b'$ the latest
(if $v$ is not a $g_\eps$-normal to $\tilde{S}$ it must stop maximising the distance immediately
(cf.\ Remark~\ref{rem4.5} (ii) below), if it is a null normal Lemma~\ref{GS_refine} applies).

Now assume that the $g$-null geodesic $\gamma $ maximises the distance to
$\bar{U}\cap \tilde S$ until the parameter value $b$. We then proceed in parallel
to the final part of the proof of Theorem~\ref{3.4}: Let $b''$ be such that $b'<b''<b$ and set $q:=\gamma(b'')$. There exist points $q_k \in
\partial J^+_k(\bar{U}\cap\tilde S)$ with $q_k \to q$. By Proposition~\ref{prop:4},
since $q_k \in \partial J^+_k (\bar{U}\cap \tilde S)$ there
exists a past directed $\check g_{\eps_k}$-null geodesic starting at $q_k$ that is
contained in $\partial J^+_k (\bar{U}\cap \tilde S)$ and is either past
inextendible or ends in $\bar{U}\cap \tilde S$. Again let $\gamma_k \colon I_k \to M$ denote
an inextendible future directed reparametrisation of such a geodesic, this time
with $\gamma_k(b'')=q_k$ and $\|\dot{ \gamma}_k(b'')\|_h=\|\dot{\gamma}(b'')\|_h$. As in Theorem
\ref{3.4} we may assume that $\dot{\gamma}_k(b'')$ converges to a $g$-null vector $v$
and that the $\gamma_k$ converge to the corresponding geodesic $\gamma_v$ in $C^1$.

For each $k$ there either exists some $0<t_k < b''$ with $\gamma_k(t_k)\in
\bar{U}\cap \tilde S$ and $\gamma_k|_{[t_k,b'']}\sse \partial J^+_k(\bar{U}\cap \tilde
S)$, or $\gamma_k|_{[0,b'']}\subseteq \partial J^+_k(\bar U \cap \tilde S)$. In the second
case we set $t_k=0$, to obtain a sequence that without loss of generality converges to some
$t'$ and $t'= 0 <b''$ or $\gamma_v (t')\in \bar U \cap \tilde S$. Since $q\not \in \bar U \cap \tilde
S$ the second case also gives $t'<b''$ and there exists $t'<t''<b''$ such that $\gamma_k|_{[t'',b'']}\sse \partial
J^+_k(\bar U \cap \tilde S) \sse \overline{J^+(\bar U \cap \tilde S)}$ for large
$k$. Consequently,
$\gamma_v|_{[t'',b'']}\sse \overline{J^+(\bar U \cap \tilde S)}$, and as in Theorem~\ref{3.4} this
implies that $\gamma=\gamma_v$.

We now note that by shrinking $U$ we may assume that
$\gamma$ can only intersect $\bar U \cap \tilde S$ once: in fact, we may locally
view $\tilde S\cap \bar U$ as a submanifold of some spacelike hypersurface $\hat
S$. By~\cite[Lemma A.25]{hawkingc11}, there exists an open set $W$ in $M$ such
that $W\cap\hat S$ is a Cauchy hypersurface in $W$. Also, since $M$ is strongly
causal, $W$ can be chosen in such a way that
$\gamma$ can only intersect it once by Lemma~\ref{lem:18 (strong causality equiv)}.

Consequently, we must have $t'=0$. Since $\gamma_k|_{[t_k,b'']} \sse \partial J^+_k(\bar U
\cap \tilde S)$, any such segment must be maximising for $\check g_{\eps_k}$.
For $k$
large we have $\dot{\gamma }_k(t_k)\in V$ since $\gamma_k \to \gamma$. Therefore, by what was shown above, $\gamma_k$ must
stop maximising the distance to $\bar U \cap \tilde S$ already at $t=t_k+b'<b''$, a
contradiction.
\end{proof}

\begin{remark} \label{rem4.5}
\begin{inlineenum} \item
In case $m=2$ (i.e., the traditional trapped surface case)
a slightly perturbed version of ~\eqref{RicciArbCodim} (namely with right hand side
$-\delta$ for any given $\delta>0$) is automatically satisfied if the null convergence condition holds: Choose $e_{n-1},e_n$ such that $e_n$ is timelike, $\dot{\gamma}(0)=e_{n-1}+e_n$
and $e_1,\dots,e_n$ is an orthonormal basis and denote the parallel translates of
$e_1,\dots,e_n$ along $\gamma$ by $E_1,\dots ,E_{n}$. Now let $\bar{E}_1,\dots ,\bar{E}_n$ be arbitrary continuous extensions of $E_1,\dots ,E_{n}$ to a neighbourhood $U$ of $\gamma$
and set $\bar{N}=\bar{E}_{n-1}+\bar{E}_n$.

Cover $\gamma$ by finitely many totally normal neighbourhoods. Then in each such neighbourhood $V$
we may parallelly transport $E_1,\dots ,E_{n}$ from some point of $\gamma$ in $V$ radially outward
to obtain local orthonormal fields $\tilde E_1,\dots,\tilde E_{n}$, and $\tilde N=\tilde{E}_{n-1}+\tilde{E}_n$. Then $\sum_{i=1}^{n-2} \langle R(\tilde{E}_i,\tilde{N})\tilde{N},\tilde{E}_i \rangle = \Ric (\tilde{N},\tilde{N})\ge 0$ on $V$. Now, as in section~\ref{reg_sec},
shrinking $U$ produces~\eqref{RicciArbCodim} with right hand side negative but arbitrarily close
to $0$. The proof of Proposition~\ref{notmaxtosupportsubmf} then still gives the desired result.

\item If $v\in TM|_{\tilde{S}}$ is future directed causal, but not a null normal to
$\tilde{S}$, then $\gamma_v $ enters $I^+(\tilde{S}) $ immediately: This is well known
for smooth metrics (\cite[Lem.~10.50]{ON83}). If $g$ is only $C^{1,1}$ one cannot use
the exponential map to construct a $C^2$-variation with a given variational vector field, but since this is a local question (and clearly true if $v$ is timelike) we may assume that $M=\R ^n$, $\gamma_v (0)=0$ and $v$ is null. We now
construct suitable variations as follows: Since $v\notin T_0\tilde{S}^\perp$ there
exists $y\in T_0\tilde{S}$ such that $\langle y, v \rangle_g >0$. Let $\alpha
 \colon [0,b]\to \tilde{S}$ be a $C^2$-curve with $\dot{\alpha}(0)=y$ (and $\alpha (0)=0$). We define a $C^2$-variation $\sigma \colon [0,t_0]\times [0,s_0] \to \R^n$ by $\sigma (t,s):= \gamma_v(t)+(1-\frac{t}{t_0})\alpha(s)$. Now let $t_0, s_0>0$ be small enough such
that $\langle y, \dot{\gamma}_v(t) \rangle_{g(\sigma(t,s))} >c>0$ for all $t\leq t_0$ and $s\leq s_0$. We will show that $\sigma(.,s)$ is a timelike curve for small $s$ and $t_0$, proving the claim.
Expanding $\alpha(s)$ and $g(\sigma(t,s))$ in a Taylor series around $s=0$ gives $\alpha(s)=sy+O(s^2)$ and $|g(\sigma(t,s))-g(\gamma_v(t))|\leq C s(1-\frac{t}{t_0})+O(s^2)$ (where $C>0$ does not depend on $s,t$) as $s\to 0$ and thus
 \begin{equation*}
  \begin{split}
     \langle \partial_t \sigma(t,s),\partial_t \sigma(t,s) \rangle_{g(\sigma(t,s))} &=\langle \dot{\gamma}_v(t),\dot{\gamma}_v(t) \rangle_{g(\sigma(t,s))}-2\frac{s}{t_0} \langle \dot{\gamma}_v(t),y \rangle_{g(\sigma(t,s))}+O(s^2)\\
   & \le s \left(\tilde{C}\left(1-\frac{t}{t_0}\right)-c\frac{2}{t_0}\right)+O(s^2).
  \end{split}
 \end{equation*}
The bracketed term evidently is negative for small $t_0$ and thus for such $t_0$ the curve $t\mapsto \sigma(t,s)$ will be a timelike curve from $0$ to $\gamma_v(t_0)$ for small $s$.
\end{inlineenum}
\end{remark}

\begin{Proposition} \label{trappedsubmfistrapped}
Let $(M, g)$ be a strongly causal $C^{1, 1}$-spacetime and let $S$ be a ($C^0$) trapped submanifold of co-dimension $1<m<n$ such that, if $m \neq 2$, the support submanifolds $\tilde S$ from Definition~\ref{support_submf} satisfy~\eqref{RicciArbCodim} for all null normals and, if $m = 2$, the null convergence condition is satisfied.
Then $E^+(S)$ is compact or $M$ is null geodesically incomplete.
\end{Proposition}

\begin{proof} Assume $M$ is null geodesically complete and fix a Riemannian
 metric $h$ on $M$ and let $K:=\{v\in TM|_S : v\, \mathrm{future\,
  directed,\, null},\, \|v\|_h=1\}$. Clearly $K$ is compact and by
 Proposition~\ref{notmaxtosupportsubmf} and Remark~\ref{rem4.5} for any $v\in K$ there exists a time $t_v$ such
 that
 $\exp (t_v\, v)\in I^+(\tilde{S})\sse I^+(S)$. Since $(v,t)\mapsto \exp
 (tv)$ is continuous there even exists a neighbourhood $U_v$ such that $\exp
 (t_v
 \, w)\in I^+(S)$ for all $w\in U_v$. By compactness we may cover $K$
 by
 finitely many of these $U_v$ and thus there exists $T$ such that
 $E^+(S)\sse
 \exp ( [0,T] \cdot K)$. This shows that $E^+(S)$ is relatively
 compact.

It remains to show that $E^+(S)$ is closed. Let $p_i=\exp (t_i v_i )\in E^+(S)$
be a sequence with $p_i\to p$ for some $p\in M$. Clearly $p\notin I^+(S)$, so
it remains to show that $p\in J^+ (S)$. Since $t_i \leq T$ and $v_i\in K$
we may assume that $t_i\to t$ and $v_i\to v\in K$. But then
since $p_i\in E^+(S)$ we must have $t_i\le t_v$ for $i$ large, hence
$p=\exp (tv)\in J^+ (S)$ and we are done.
\end{proof}

\begin{Corollary}
\label{achronalnotnecessary}
Let $(M,g)$ and $S$ be as in the previous proposition. Then $E^+(S)\cap S$ is an achronal set and $E^+(E^+(S)\cap S)$ is compact or $M$ is null geodesically incomplete.
\end{Corollary}

\begin{proof}
This follows verbatim as in the smooth case, see~\cite[Prop.\ 4]{GS} or~\cite[Prop.\ 4.3]{Seno1}, using that by definition for any $p\in S$ there exists a neighbourhood $U_p$ such that $S\cap U_p$ is achronal in $U_p$.
\end{proof}

\subsection{Trapped points}\label{sec:trappedpt}

In the classical smooth version of the Hawking--Penrose theorem there is a third initial condition concerning a `trapped point' $p$, which is a point $p$ such that the expansion becomes negative for any future directed null geodesic starting in $p$. This condition can again be formulated in a precise way in the language of Jacobi tensors, see e.g.~\cite[Prop.\ 12.46]{BEE}, by demanding that for any future directed null geodesic $\gamma $ starting in $p$ the expansion $\theta (t)$ associated to the unique Jacobi tensor class $[A]$ along $\gamma $ with $[A](0)=0$ and $[\dot{A}](0)= \mathrm{id}$ becomes negative for some $t>0$. This formulation unfortunately does not generalise to a $C^{1,1}$-metric (one of the reasons for this being that there is no sensible way to formulate the Jacobi equation). There is, however, an equivalent formulation for smooth metrics using a shape operator of spacelike slices of the lightcone of $p$ (which is similar to the use of co-spacelike distance functions and their level sets in the timelike or Riemannian case, cf.~\cite[Appendix~B.3]{BEE}):

Let $\gamma $ be a null geodesic and assume that the expansion of the Jacobi tensor class $[A]$ along $\gamma $ with $[A](0)=0$ and $[\dot{A}](0)= \mathrm{id}$ becomes negative for some $t>0$. We set $t_0:=\inf \{t>\eta :\,\theta(t)<0\} $, where $\eta >0$ is chosen such that $[0,\eta ]\cdot \dot{\gamma}(0)$ is contained in a neighbourhood where $\exp_p $ is a diffeomorphism. This ensures that $\gamma (t_0)$ must come before the first conjugate point of $p$ and so there exists $t_1>t_0$ such that $\gamma |_{[0,t_1]}$ does not contain points conjugate to $p$ along $\gamma $. Thus, there exists a neighbourhood $U \sse T_pM$ of $[0,t_1]\cdot \dot{\gamma}(0)$ such that $\exp_p|_{U}$ is a diffeomorphism onto its image: It clearly is a local diffeomorphism and if it were not injective on any such neighbourhood there would exist vectors $X_k,Y_k\in T_pM$, $X_k \neq Y_k$, converging to $X,Y\in [0,t_1]\cdot \dot{\gamma }(0)$ with $\exp_p(X_k)=\exp_p(Y_k)$, hence $\exp_p(X)=\exp_p(Y)$. Since $\exp_p $ is locally injective $X\neq Y$ but this contradicts $\exp_p $ being injective on $[0,t_1]\cdot \dot{\gamma }(0)$ by causality of $M$.

Now, one can look at the level sets $S_t:=\exp_p(t\,\tilde{U})$, where $\tilde{U}:=\{v
\in U: v\,\mathrm{null},\, g(T,v)=g(\dot{\gamma}(0),T)\}$ for some fixed timelike vector $T\in T_pM$, and their shape operators $\mathbf{S}_{\dot{\gamma}(t)}(t) \colon T_{\gamma(t)}S_t \to T_{\gamma(t)}S_t$ derived from the normal $\dot{\gamma}(t)$. Proceeding as in~\cite[Prop.\ 3.4]{Grant11} one gets that
this shape operator satisfies a Riccati equation along $\gamma $ and $\lim_{t\searrow 0}
t \, \mathbf{S}_{\dot{\gamma}(t)}(t) =\mathrm{id}$. Identifying $T_{\gamma(t)}S_t $ with $[\dot\gamma(t)]^\perp$,
a quick calculation shows that the tensor class $[B]$ along $\gamma $ defined by
$[\dot{B}]= \mathbf{S}_{\dot{\gamma}}[B]$ on $(0,t_1)$ and $[B](t_0)=[A](t_0)$ also satisfies the Jacobi equation
and hence can uniquely be extended to $(-\infty,\infty)$. From the limiting
behaviour of $\mathbf{S}_{\dot{\gamma}(t)}(t)$ as $t\searrow 0$ one gets $[B](0)=0$ and thus by uniqueness of Jacobi tensors $[B]=[A]$ on
$[0,t_1)$, so $\mathbf{S}_{\dot{\gamma}(t)}(t) =[\dot{A}](t)[A]^{-1}(t)$ and $\theta (t)=\mathrm{tr} \,
\mathbf{S}_{\dot{\gamma}(t)}(t)$ for $t<t_1$. Consequently, a negative $\theta (t)$ corresponds to a negative trace of
the shape operator of the spacelike surface $S_t$ with respect to the normal $\dot{\gamma} $.
Since $\mathbf{k}_{S_t}(\dot{\gamma}(t))= -\mathrm{tr} \, \mathbf{S}_{\dot{\gamma}(t)}(t)$ this is equivalent
to $\mathbf{k}_{S_t}(\dot{\gamma}(t))$ being positive.

This condition can now be generalised to $C^{1,1}$-metrics and, as introduced in section~\ref{sec:mainresult}, we give the following definition of a (future) trapped point. Note that this can very roughly be seen as a condition on the mean curvature of the level set $S_t$ (which is now at best Lipschitz) in the sense of support submanifolds and hence bears some similarities to our definition of past-pointing timelike mean curvature for $C^0$-submanifolds.

\begin{definition}\label{def:trappedpt}
We say that a point $p$ is future trapped if for any future-pointing null vector $v \in T_pM$ there exists a $t$ such that there exists a spacelike $C^2$-surface $\tilde{S}\sse J^+(p)$ with $\gamma_v(t)\in \tilde{S}$ and $\mathbf{k}_{\tilde{S}}(\dot{\gamma}_v(t))>0$.
\end{definition}
Using this definition one can easily prove that $E^+(p)$ is compact for a trapped point $p$.
\begin{Proposition}
\label{trappedptistrapped}
Let $(M, g)$ be a strongly causal $C^{1,1}$-spacetime and assume that the null convergence condition
holds. If $p \in M$ is a future trapped point and $M$ is null geodesically complete then $E^+(p)$ is compact.
\end{Proposition}
\begin{proof}
The proof is completely analogous to the one of Proposition~\ref{trappedsubmfistrapped}, using that $\tilde{S}$ is a surface and thus condition~\eqref{RicciArbCodim} is not required if the null convergence condition holds (cf.\ Remark~\ref{rem4.5}).
\end{proof}
\section{Proof of the main result}
\label{sec:proof}

As in the smooth case we will first prove a $C^{1,1}$-version of Theorem~\ref{C2HPCausalityVersion}.
To do so, we will roughly follow the original proof in~\cite{HP}. However, we will split the argument into smaller
pieces to better highlight the places where the reduced regularity of the metric has to be taken into account.
In an attempt to keep our presentation concise we start only with the proof of~\cite[Lemma~2.12]{HP} (which will be
Corollary~\ref{cor:27} here), but for completeness all necessary preliminary results are collected in the appendix.
Our notation in this section follows, e.g.,~\cite{ON83}, but is also explicitly defined in the introduction or the
appendix. In what follows we always assume $S$ to be non-empty.

\begin{Lemma}\label{prop:25}
Let  $(M,g)$ be a spacetime with a $C^{1,1}$-metric $g$, let $S$ be an achronal and closed subset of $M$ and suppose
that strong causality holds on $M$. Then $H^+(\overline{E^+(S)})$ is non-compact
or empty.
\end{Lemma}
\begin{proof}
The proof is completely analogous to the smooth one found in, e.g.,~\cite[Lemma~9.3.2]{Krie}. Note that Lemma 9.3.1 and Lemma 8.3.8
from that reference still hold (see Corollary~\ref{cor:22} and Proposition~\ref{prop:10}) and that the curve $\beta_1$, which starts outside of $D^+(\overline{E^+(S)})$ and ends in $S$, must intersect $H^+(\overline{E^+(S)})$ by Lemma~\ref{lem:12}.
\end{proof}

\begin{Corollary}
\footnote{cf.~\cite[Lemma~2.12]{HP}}
\label{cor:27}
Let $(M,g)$ be a spacetime with a $C^{1,1}$-metric $g$ that is strongly causal. Let $S\sse M$ be an achronal set and assume that $E^+(S)$ is compact.
Then there exists a future-inextendible timelike curve $\gamma$ contained
in $ \dpep ^\circ$.
\end{Corollary}
\begin{proof}
The proof is completely analogous to the smooth case,~\cite[Lemma~2.12]{HP}.
By Lemma~\ref{lem:7} we may assume that
$S$ is closed.
The idea is that, if every timelike curve that meets
$E^+(S)$ also meets $\hpep$ (or equivalently leaves $\dpep^\circ$), then, using that $\hpjp $ is a topological hypersurface by Lemma~\ref{lem:16}, one can define a continuous map from
$E^+(S)$ to $\hpep$ via
the flow of a smooth timelike vector field. This gives a contradiction since $E^+(S)$ is non-empty
and compact but $\hpep$ is empty or non-compact by Lemma~\ref{prop:25}.
\end{proof}

The next Lemma will extract the part of the proof of Theorem~\ref{C11HPCausalitybit}, where the original proof (and also the one in~\cite{Seno1}) argues using the continuous dependence of conjugate points on the geodesic, which is evidently a problem for $C^{1,1}$-metrics. There are, however, smooth proofs that avoid this, see e.g.~\cite[Lemma 9.3.4]{Krie}. While that proof should also work in $C^{1,1}$ (and we will refer to parts of it), we will still present a different argument of the crucial step more in line with the original proof.

\begin{Lemma}
\footnote{cf.~\cite[pp.~545]{HP}.}
\label{prop:28}
Let $(M,g)$ be a spacetime with a $C^{1,1}$-metric $g$ that is strongly causal and assume that no inextendible null geodesic in $M$ is globally maximising. Let $S$ be achronal and assume that $E^+(S)$ is compact, and let $\gamma$ be a future inextendible timelike curve contained in $D^+\left(E^+(S)\right)^\circ$. Then $F := E^+(S) \cap \overline{J^{-}\left(\gamma\right)}$ is achronal and $E^-(F)$ is compact.
\end{Lemma}

\begin{proof}
By Lemma~\ref{lem:7} we may without loss of generality assume that $S$ is closed.
Since $F \subseteq E^+(S)$ and $E^+(S)$ is achronal, it follows that $F$ is achronal. Moreover, $E^+(S)$ is, by assumption, compact and $\overline{J^{-}\left(\gamma\right)}$ is closed, therefore $F$ is compact. We need to show that $E^{-}(F)$ is compact.
To do so, first note that the same arguments as in~\cite[Lemma~9.3.4]{Krie} show that
\begin{equation}
E^{-}(F) \sse F \cup \partial J^{-}(\gamma).
\label{eq:E-(F) without F subset jp gamma}
\end{equation}
	
	Now let $v \in TM|_{F}$ be past pointing causal. Then, by the definition of $F$, the past inextendible geodesic $c_{v} \colon [0, b) \to M$
	with initial velocity $\dot c_{v}(0)=v$ must be contained in $\overline{J^{-}(\gamma)}$. We show that $c_{v} \cap I^{-}(\gamma) \neq \emptyset$:
	If $c_{v}$ never met $I^{-}(\gamma)$ it would have to be a null
	geodesic and lie entirely in $\partial J^{-}(\gamma) \setminus E^{-}(\gamma)$
	(since $E^{-}(\gamma) = \emptyset$ because $\gamma $ is future inextendible timelike). In particular
	$c_{v}(0)\in\partial J^{-}(\gamma)\setminus E^{-}(\gamma)$, so by
	Proposition~\ref{prop:4}
(note that the image of $\gamma$ is a closed set by Lemma~\ref{lem: 28 the new one}),
	there exists a future directed, future inextendible null geodesic $\lambda$ that starts
	at $c_{v}(0)$ and is contained in $\partial J^{-}(\gamma)$. But
	then $c_{v}\lambda$ either is an inextendible broken null geodesic, hence
	not maximizing by Lemma~\ref{pu:C11}, or it is an inextendible unbroken null geodesic, hence
	not maximizing by assumption.
	Hence by Lemma~\ref{pu}, $c_{v}\lambda$ cannot lie entirely in $\partial J^{-}(\gamma)$,
	giving a contradiction.
	Consequently, for all $v \in TM|_F$, there exists a $t_v$ with $c_v(t_v)\in I^{-}(\gamma)$. Since $I^-(\gamma )$ is open there exists a neighbourhood $U_v \subseteq TM$ of $v$ such that $c_w$ is defined on $[0,t_v)$ and $c_w(t_v)\in I^-(\gamma )$ for all $w\in U_v$. By compactness of $F$ one can cover the set of all $h$-unit, past pointing causal vectors in $TM|_F$ by finitely many of these neighbourhoods, which shows that $E^-(F)\cap \partial J^{-}(\gamma)$ is relatively compact. In fact, it is actually compact as can easily be seen using a limit argument as in the final part of the proof of Proposition~\ref{trappedsubmfistrapped} (which does not use null completeness).
	This shows that $E^-(F)$ is compact by~\eqref{eq:E-(F) without F subset jp gamma} and compactness of $F$.
\end{proof}

Combining these preliminary results allows us to prove the low-regularity version of Theorem~\ref{C2HPCausalityVersion}.
Again the argument proceeds very similarly to the smooth case, but we nevertheless give a complete proof.

\begin{Theorem}
\label{C11HPCausalitybit}
Let $(M,g)$ be a spacetime with a $C^{1,1}$-metric $g$. Then the following four conditions cannot all hold:
\begin{enumerate}[label={(C.\roman*)}, ref={C.\roman*}]
\item\label{thm:33:1} $M$ contains no closed timelike curves;
\item\label{thm:33:2} Every inextendible timelike geodesic contained in an open globally hyperbolic subset stops being maximizing;
\item\label{thm:33:3} Every inextendible null geodesic stops being maximizing;
\item\label{thm:33:4} There is an achronal set $S$ such that $E^+(S)$ or $E^-(S)$ is compact.
\end{enumerate}
\end{Theorem}
\begin{proof}
We assume, to the contrary, that all four conditions hold. From conditions~(\ref{thm:33:3}) and~(\ref{thm:33:1}), Lemma~\ref{lem: 19 conj. pts imply strong causality} implies that $M$ is strongly causal. In condition~(\ref{thm:33:4}) we assume, without loss of generality, that $E^+(S)$ is compact.

Let $\gamma$ be a future inextendible timelike curve contained in $D^+\left(E^+(S)\right)^\circ$ given by Corollary~\ref{cor:27}, and let $F := E^+(S) \cap \overline{J^{-}(\gamma)}$ as in Lemma~\ref{prop:28}. Then, by Lemma~\ref{prop:28}, the set $F$ is achronal and $E^-(F)$ is compact. Therefore, by Corollary~\ref{cor:27}, there exists a past-inextendible timelike curve $\lambda$ contained in the set $D^{-}(E^{-}(F))^\circ$.

Next we show $\gamma \subseteq D^{+}(E^{-}(F))^\circ$: We have $\gamma \subseteq D^+\left(E^+(S)\right)^\circ$, so every past
inextendible causal curve starting at $\gamma$ must meet $E^+(S)$. This meeting point is obviously in $J^{-}(\gamma) \subseteq \overline{J^{-}(\gamma)}$, so every past inextendible causal curve starting at $\gamma$ meets $E^+(S) \cap \overline{J^{-}(\gamma)} = F$, which gives $\gamma \subseteq D^{+}(F)$. Also $\gamma$ cannot meet $\partial D^{+}(F)$, which is equal to $F \cup H^{+}(F)$ by Proposition~\ref{prop:10}, since $F \subseteq E^+(S)$ and $\gamma \subseteq D^{+}(E^+(S))^\circ$ and if $\gamma$ met $H^{+}(F)$ it would also meet $I^{+}(H^{+}(F))$ by being timelike,
hence leave $D^{+}(F)$ (by Lemma~\ref{lem:15}). This means that $\gamma \subseteq D^{+}(F)^\circ \subseteq D^{+}(E^{-}(F))^\circ$ (by achronality of $F$).

So both $\gamma$ and $\lambda$ are contained in $D(E^{-}(F))^\circ$. By~\cite[Thm.~A.22]{hawkingc11}, $D(E^{-}(F))^\circ$ is globally hyperbolic. Now choose sequences $\left\{ p_{k}\right\} \subseteq \lambda$ and $\left\{ q_{k}\right\} \subseteq \gamma$ with the following properties:
\begin{enumerate}
\item $p_{k+1}\in I^{-}(p_{k})$ and $q_{k+1}\in I^{+}(q_{k})$,
\item both $\left\{ p_{k}\right\} $ and $\left\{ q_{k}\right\} $ leave every compact subset of $M$,
and
\item $q_{1}\in I^{+}(p_{1})$. To see that this is possible, note that $\lambda \sse J^-(E^-(F))\sse J^-(F) \sse J^-(\overline{J^-(\gamma )} )$ and since $\lambda $ is timelike Lemma~\ref{pu} gives that $\lambda \sse I^-(\overline{J^-(\gamma )})=I^-(\gamma)$.
\end{enumerate}
By~\cite[Prop.~6.4]{S14}
there exist maximizing causal curves $\gamma_{k} \colon [a_k,b_k]\to M$ from $p_{k}$ to $q_{k}$. Each $\gamma_{k}$ must intersect $E^{-}(F)$
(because it connects $D^{-}(E^{-}(F))^{\circ}$ with $D^{+}(E^{-}(F))^{\circ}$, cf.\ the remark preceding~\cite{ON83}, Lemma 14.37) in some point $r_{k}$. By compactness of $E^{-}(F)$ (see Lemma~\ref{prop:28}) we may assume that $r_{k}\to r$ after passing to a subsequence if necessary, so there exists a causal limit curve $\tilde\gamma$ by Theorem~\ref{prop:MinguzziLimit}.

Now because every $\gamma_{k}$ is maximising the sequence $\left\{ \gamma_{k}\right\} $ is limit maximising in the sense of~\cite[Def.~2.11]{Minguzzi_LimitCurveThms} and thus
$\tilde\gamma$ has to be maximising (again by Theorem~\ref{prop:MinguzziLimit}).
Also, since $\left\{ p_{k}\right\} $ and $\left\{ q_{k}\right\} $ leave every compact set, $\tilde\gamma$ is inextendible. Because $\tilde\gamma$ is maximising it has to be a geodesic (cf.~\cite[Thm.~1.23]{M}).

If $\tilde\gamma$ is null this immediately contradicts the third assumption and we are done. Since $D(E^-(F))^\circ $ is globally hyperbolic, to establish a contradiction to condition~\ref{thm:33:2} it only remains to show that $\tilde\gamma \subseteq D(E^-(F))^\circ $ if it is timelike. Since it is the limit of the $\gamma_k$'s we certainly have $\tilde\gamma \subseteq \overline{D(E^-(F))} $. Now, Proposition~\ref{prop:10} implies
\[
\partial D(E^-(F)) \subseteq H^+(E^-(F))\cup E^-(F) \cup H^-(E^-(F)).
\]
Since $I^+(H^+(E^-(F)))=I^+(E^-(F))\setminus \overline{D^+(E^-(F))}$ (see Lemma~\ref{lem:15}) it follows that $\tilde\gamma \cap H^+(E^-(F)) = \emptyset $ and, analogously, $\tilde\gamma \cap H^-(E^-(F)) = \emptyset $. Now assume there exists $t_0$ such that $\tilde\gamma(t_0)\in E^-(F)$. We show that then $\tilde\gamma(t_0)\in D(E^-(F))^\circ $. By achronality of $E^-(F)$ and the above we get \[ \tilde\gamma(t_0+1)\in \overline{D^+(E^-(F))}\setminus (E^-(F)\cup H^+(E^-(F)))=D^+(E^-(F))^\circ \subseteq D(E^-(F))^\circ \] and by the same argument also $\tilde\gamma(t_0-1) \in D(E^-(F))^\circ $. But then since $D(E^-(F))^\circ $ is globally hyperbolic we have that the causal diamond $J(\tilde\gamma(t_0-1),\gamma_v(t_0+1)) \subseteq D(E^-(F))^\circ$ and hence $\tilde\gamma(t_0)\in D(E^-(F))^\circ$.
\end{proof}

Collecting this and the results established in the previous sections, we are now in the position to prove Theorem~\ref{HPinC1,1} and Theorem~\ref{ArbitrCodiminC1,1}.

\subsection*{Proof of the Hawking--Penrose Theorem for $C^{1,1}$-metrics}
\begin{proof}
We show that, for a causally geodesically complete spacetime $(M, g)$, assumptions~\ref{HPinC1,1:1} to~\ref{HPinC1,1:4} in Theorem~\ref{HPinC1,1} and Theorem~\ref{ArbitrCodiminC1,1} imply that conditions~(\ref{thm:33:1}) to~(\ref{thm:33:4}) of Theorem~\ref{C11HPCausalitybit} are satisfied.

Clearly, causality is a stronger assumption than being chronological, so~\ref{HPinC1,1:1} implies~(\ref{thm:33:1}). Theorem~\ref{timelikenotmax} shows that the strong energy and the genericity conditions (i.e.\ assumptions~\ref{HPinC1,1:2} and~\ref{HPinC1,1:3} of Theorem~\ref{HPinC1,1}) imply that condition~(\ref{thm:33:2}) of Theorem~\ref{C11HPCausalitybit} is satisfied. Similarly, Theorem~\ref{3.4} shows that assumptions~\ref{HPinC1,1:1},~\ref{HPinC1,1:2} and~\ref{HPinC1,1:3} of Theorem~\ref{HPinC1,1} imply that condition~(\ref{thm:33:3}) of Theorem~\ref{C11HPCausalitybit} holds.

Finally, Proposition~\ref{hypersurfacecase} shows that assumption~\ref{HPinC1,1:4i} implies condition~(\ref{thm:33:4}). Since we have already established that conditions~(\ref{thm:33:1}) and (\ref{thm:33:3}) of Theorem~\ref{C11HPCausalitybit} hold, Lemma~\ref{lem: 19 conj. pts imply strong causality} in the appendix implies that $(M, g)$ is strongly causal. Therefore, one can apply Proposition~\ref{trappedsubmfistrapped} (with Corollary~\ref{achronalnotnecessary}) and Proposition~\ref{trappedptistrapped} to show that any one of the assumptions~\ref{HPinC1,1:4ii},~\ref{HPinC1,1:4iii} or~\ref{HPinC1,1:4iv} (together with assumptions~\ref{HPinC1,1:1}--\ref{HPinC1,1:3}), implies that condition~(\ref{thm:33:4}) of Theorem~\ref{C11HPCausalitybit} is satisfied.
\end{proof}
\appendix
\section{Causality results in $C^{1,1}$}
\label{app:C11causality}

Standard expositions of causality theory (\cite{HE,Seno1,Seno2,Chrusciel_causality,ladder}) usually assume the metric to be
at least $C^2$. Most results, however,
remain true for $C^{1,1}$-metrics, see~\cite{CG,M,KSSV} and the appendix of~\cite{hawkingc11}. In this appendix we will collect further results that are not included in these previous works, but are necessary for the proof of Theorem~\ref{C11HPCausalitybit}.

In the following we will always assume that $(M,g)$ is a spacetime with a $C^{1,1}$-metric unless explicitly stated otherwise. We also fix a smooth Riemannian background metric $h$.

\subsection{Limit curves and the structure of $\partial J^+(S)$}

Two important results from~\cite{CG} are that $I^\pm(S)$ is open (\cite[Prop.~1.21]{CG}) and that the push-up principle remains true (\cite[Lem.~1.22]{CG}) for causally plain spacetimes. As these include the class of spacetimes with Lipschitz
continuous metrics (\cite[Cor.~1.17]{CG}), we have

\begin{Lemma} Let $S\sse M$. Then $I^\pm(S)$ is open.
\end{Lemma}

\begin{Lemma}
\label{pu}
Let $p, q, r \in M$ be such that $p\leq q \ll r$ or $p \ll q \leq r$. Then $p\ll r$.
\end{Lemma}

We will also repeatedly be making use of the following result, see~\cite[Lem.\ 2]{M}:
\begin{Lemma}
\label{pu:C11}
Let $p, q \in M$ such that there exists a future directed causal curve $c$ from $p$ to $q$. Then either $q \in I^+(p)$ or $c$ is (can be reparametrised to) a maximising null geodesic from $p$ to $q$.
\end{Lemma}

Using the usual notation, we set $E^+(S):=J^+(S)\setminus I^+(S)$. It is easily checked that (as for smooth metrics) we have:
\begin{Lemma}
\label{newlem}
Let $S \subseteq M$.
Then both $\ep$ and $\jp$ are achronal sets, $\jp$ is closed, but $\ep$ need not be.
\end{Lemma}

\begin{Lemma}
\label{lem:3}
Let $S\sse M$. Then $\jp$ is an achronal, closed topological hypersurface.
\end{Lemma}
\begin{proof}
Clearly $J^{+}\left(J^{+}(S)\right)=J^{+}(S)$, so~\cite[Corollary 14.27]{ON83}, which is easily verified to hold for $C^{1,1}$-metrics as well, gives the desired result.
\end{proof}

To proceed further we are going to need some results on limits of causal curves.
Thus we will now state that what is essentially Theorem 3.1.(1) from~\cite{Minguzzi_LimitCurveThms} remains true for $C^{1,1}$-metrics.

\begin{Theorem}\label{prop:MinguzziLimit}
Let $y$ be an accumulation point of a sequence of (future directed) causal curves. There is a subsequence parametrized with respect to $h$-length, $\gamma_k \colon [a_k,b_k]\to M$ ($a_k$ and $b_k$ may be infinite), $0\in [a_k,b_k]$ such that $\gamma_k(0)\to y$ and such that the following
properties hold. There are $a\leq 0$ and $b\geq 0$, such that $a_k\to a$ and $b_k\to b$.
If there is a neighbourhood $U$ of $y$ such that only a finite number of $\gamma_k$ is
entirely contained in $U$ then there is a causal
curve $\gamma \colon [a,b]\to M$, such that $\gamma_k$
converges $h$-uniformly on compact subsets to $\gamma $. This limit curve is past, respectively future, inextendible if and only if $a=-\infty$, respectively $b=\infty$. Further, if $\gamma_k$ is limit maximising (in the sense of~\cite[Def.~2.11]{Minguzzi_LimitCurveThms}) then $\gamma $ is maximising.
\end{Theorem}
\begin{proof}
The existence of such a limit curve follows from the smooth version~\cite[Thm.~3.1.(1)]{Minguzzi_LimitCurveThms} in the same way as in the proof of~\cite[Thm.~1.5]{S14}. This also immediately gives the statement about inextendibility. That the limit of a limit maximising sequence is maximising follows as in the smooth case (see~\cite[Thm.~2.13]{Minguzzi_LimitCurveThms}), using that for $C^{1,1}$-metrics the Lorentzian distance function is still lower semi-continuous (see~\cite[Lemma~A.16]{hawkingc11}) and that the length functional is still upper semi-continuous (see~\cite[Thm.~6.3]{S14} and note that it does not require the same start and end points but only a uniform bound on the Lipschitz constants).
\end{proof}

We now use this to show that as in the smooth case the boundary of the causal future $\partial J^+(S)$, is ruled by null geodesics that are either past inextendible or end in $\bar{S}$. This result is needed for the proof of both Theorem~\ref{3.4} and Proposition~\ref{notmaxtosupportsubmf}.

\begin{Proposition}
\label{prop:4} Let $S\sse M$.
Any $x \in \jp \setminus \bar{S}$ is the
future end point of a causal curve $\gamma \subseteq \jp$ that either
is past inextendible (and never meets $\bar{S}$) or has a past endpoint
in $\bar{S}$. This $\gamma$ is (can be reparametrised to) a maximising
null geodesic. If $S$ is closed and $x\notin J^+(S)$, then this curve is past inextendible and contained in $\jp \setminus J^+(S)$.
\end{Proposition}
\begin{proof}
Let $x \in \jp \setminus \bar{S}$. Then there exists a sequence $\left\{ x_{k}\right\} \subset I^{+}\left(S\right)$
with $x_{k}\to x$ and past directed timelike curves $\gamma_{k} \colon \left[0,b_{k}\right] \to M$
from $\gamma_{k}\left(0\right)=x_{k}$ to $\gamma_{k}\left(b_{k}\right) \in S$.
Since $x\notin \bar{S}$ the $\gamma_{k}$'s leave a
fixed neighbourhood of $x$ and so by Theorem~\ref{prop:MinguzziLimit} there exists (a subsequence with) a limit curve $\gamma$ with
$\gamma\left(0\right)=x$ that is either past inextendible or $b_{k}\to b< \infty$
and $\gamma\left(b\right)=\lim\gamma_{k}\left(b_{k}\right) \in \bar{S}$. Clearly, $\gamma \sse \overline{J^{+}(S)}$. If $\gamma$ were ever
in $I^{+}\left(S\right)$, then $x \in I^{+}\left(S\right)$ by Lemma~\ref{pu}, a contradiction.

That $\gamma$ is (can be reparametrised to)
a maximizing null geodesic follows immediately from Lemma~\ref{pu:C11}. Finally, if $S$ is closed and $x\notin J^+(S)$ there can be no causal curve from $x$ to $S=\bar{S}$, so $\gamma$ must be inextendible and $\gamma \sse \jp \setminus J^+(S)$.
\end{proof}

\subsection{Cauchy development and Cauchy horizon}

Next, we are interested in the Cauchy developments and Cauchy horizons
of both $\ep$ and $\jp$ (and their relationship with each other). From now on we will generally require $S$ to be an achronal
(non-empty) set. Note that this implies in particular
\begin{equation}
\label{ac}
S \subseteq J^+(S) \setminus I^+(S) = E^+(S).
\end{equation}
From this one also immediately obtains the following Lemma:
\begin{Lemma}
\label{lem:7}
Let $S$ be achronal. Then $\bar{S}$ is also achronal. Further, if $\ep $ is compact, then $\ep=E^{+}(\overline{S})$.
\end{Lemma}
\begin{proof}
The first claim follows from the fact that $I^{+}(\bar{S})=I^{+}(S)$ and openness of $I^+(S)$. The same equality also immediately gives $\ep \subseteq E^{+}(\bar{S})$. Now if $\ep$ is compact, then~\eqref{ac}
implies $\bar{S} \subseteq \ep$. This gives $E^{+}(\bar{S})=J^{+}(\bar{S}) \setminus I^{+}(\bar{S}) \subseteq J^{+}(E^{+}(S)) \setminus I^{+}(\bar{S})$.
Since $J^{+}(\ep)=J^{+}(S)$ and $I^{+}(\bar{S})=I^{+}(S)$, this shows the other inclusion.
\end{proof}

\begin{definition}
Let $A$ be achronal. The \emph{future Cauchy development\/} $D^{+}(A)$ of $A$ is defined by
\footnote{We follow the convention of~\cite{HawkingIII, HE, ON83}, rather than that of~\cite{Penrose:Battelle, HP, Penrose:littlebook}.}
\begin{equation}
D^{+}(A) :=\{x \in M:\:\text{every past inextendible causal curve through }x\text{ meets }A\}
\label{eq:def D+}
\end{equation}
and its \emph{future Cauchy horizon\/} $H^{+}(A)$ is defined by
\begin{equation}
H^{+}(A) := \overline{D^{+}(A)} \setminus I^{-}\left(D^{+}(A)\right)=\left\{ x \in \overline{D^{+}(A)}:\,
I^{+}(x) \cap D^{+}(A)=\emptyset\right\} .\label{eq: def H+}
\end{equation}
\end{definition}

Two important properties of $D^{+}(A)$ for closed achronal sets $A$ are given in the following proposition.
\begin{Proposition}
\label{prop:10}
Let $A$ be closed and achronal. Then
\begin{equation}
\overline{D^{+}(A)}=\left\{ x \in M:\:\text{every past inextendible timelike curve through }x\text{ meets }A\right\} .\label{eq:closur of D+}
\end{equation}
Furthermore
\begin{equation}
\partial D^{+}(A)=A\cup H^{+}(A).\label{eq: boundary of D+}
\end{equation}
\end{Proposition}
\begin{proof}The proofs can be found in~\cite[Lemma~A.13]{hawkingc11} and~\cite[Lemma~A.14]{hawkingc11}.
\end{proof}

\begin{Lemma}
\label{prop:11}Let $A$ be closed and achronal and let $x \in D^{+}(A) \setminus H^{+}(A)$.
Then every past inextendible causal curve through $x$ must meet $I^{-}(A)$.
\end{Lemma}
\begin{proof}
Any $x \in D^{+}(A) \setminus H^{+}(A)$ is either in $A$ or in $D^{+}(A)^\circ$, so the result follows from~\cite[Lemma 8.3.6]{Krie}, which still holds for $C^{1,1}$-metrics.
\end{proof}

\begin{Lemma}
\label{lem:12}Let $A$ be closed and achronal and $x \in J^{+}(A) \setminus D^{+}(A)$ or $x\in I^+(A)\setminus D^{+}(A)^\circ$.
Then every causal curve from $x$ to $A$ must also meet $H^{+}(A)$.
\end{Lemma}
\begin{proof}
Let $x \in J^{+}(A) \setminus D^{+}(A)$ or $x\in I^+(A)\setminus D^{+}(A)^\circ$. If $x \in \overline{D^{+}(A)}$, then $x\in \partial D^{+}(A) = A\cup H^{+}(A)$ (see Proposition~\ref{prop:10}). Thus $x \in H^{+}(A)$ since in either case $x$ cannot be in $A$ because $A \sse D^+(A)$ and $I^+(A)\cap A=\emptyset$ by achronality, so we are done.

Now assume $x\notin\overline{D^{+}(A)}$ and let $\lambda$ be a causal
curve from $x$ to $A \subseteq D^{+}(A)$.
Then there exists $t_{0}>0$ such that $\lambda(t_{0}) \in \partial D^{+}(A)$
but $\lambda(t)\notin\overline{D^{+}(A)}$ for all $t<t_{0}$. We
have to show that $\lambda(t_{0}) \in H^{+}(A)$. Assume to the contrary
that $\lambda(t_{0}) \in A \setminus H^{+}(A)$ (cf.~\eqref{eq: boundary of D+}). Then $I^{+}(\lambda(t_{0})) \cap D^{+}(A)\neq\emptyset$
by definition of $H^{+}$. Now let $p \in I^{+}(\lambda(t_{0})) \cap D^{+}(A)$,
then $I^{-}(p)$ is an open neighbourhood of $\lambda(t_{0})$ so there
exists a $t_{1}<t_{0}$ such that $\lambda(t_{1})$ is still in $I^{-}(p)$.
Since $t_{1}<t_{0}$ we have $\lambda(t_{1})\notin\overline{D^{+}(A)}$,
so, by~\eqref{eq:closur of D+}, there exists a timelike past inextendible curve $\gamma$ starting at $\lambda(t_{1})$
that does not meet $A$. Concatenating any timelike curve from $p$
to $\lambda(t_{1})$ with $\gamma$ shows that this timelike curve
from $p$ to $\lambda(t_{1})$ must meet $A$ in a point that cannot be $\lambda(t_{1})$ itself (since $\lambda(t_{1})\notin\overline{D^{+}(A)}$).
But this means that $\lambda(t_{1}) \in I^{-}(A)$, giving a contradiction
to $\lambda(t_{1})\geq\lambda(t_{0}) \in A$ and achronality of $A$.
\end{proof}
We use this to give a proof of~\cite[Equation~(2.4)]{HP} in the $C^{1,1}$-setting.
\begin{Lemma}
\label{lem:15}
Let $A$ be closed and achronal. Then $I^{+}\left(H^{+}(A)\right)=I^{+}(A) \setminus \overline{D^{+}(A)}$.
\end{Lemma}
\begin{proof}
By Proposition~\ref{prop:10} we have $H^{+}(A) \subseteq \overline{D^{+}(A)} \subseteq I^{+}\left(A\right)\cup A$,
so $I^{+}\left(H^{+}\left(A\right)\right) \subseteq I^{+}\left(A\right)$.
Let $x \in I^{+}\left(H^{+}(A)\right)$ and assume $x \in \overline{D^{+}(A)}$,
then there exists a neighbourhood $U$ of $x$ such that $U \cap D^{+}\left(A\right)\neq\emptyset$
and $U \subseteq I^{+}\left(H^{+}(A)\right)$, contradicting $I^{+}\left(H^{+}(A)\right) \cap D^{+}\left(A\right)=\emptyset$
(cf.~\eqref{eq: def H+}).
So $I^{+}\left(H^{+}(A)\right) \subseteq I^{+}(A) \setminus \overline{D^{+}(A)}$.

Now let $x \in I^{+}(A) \setminus \overline{D^{+}(A)}$. Then by Lemma~\ref{lem:12} any timelike curve from $x$ to $A$ must meet $H^{+}\left(A\right)$
in some point $p$ so, since $x\notin\overline{D^{+}(A)}\supseteq H^{+}\left(A\right)$
we have $p\neq x$, and thus $x$ must be in $I^{+}\left(H^{+}(A)\right)$.
\end{proof}

\begin{Lemma}
\label{helem}
Let $S$ be closed and achronal. Then
\[
\mathrm{edge}(H^+(S)) \subseteq \mathrm{edge(S)}.
\]
\end{Lemma}
\begin{proof}
We basically follow the proof of~\cite[Prop.~6.5.2]{HE}. Let $q \in \mathrm{edge}(H^+(S))$ and let $U_k$ be a sequence of
neighbourhoods of $q$ with $U_k\to \{q\}$. By definition of edge (cf.~\cite[14.23]{ON83}), for each $n$ there
exist points $p_k \in I^-(q,U_k)$ and $r_k \in I^+(q,U_k)$ connected by a future directed
timelike curve $\lambda_k$ that does not intersect $H^+(S)$. It then follows that
$\lambda_k$ does not intersect $\overline{D^+(S)}\supseteq S$.

In particular, $r_k \in I^+(q, U_k) \subseteq I^+(q)$, so $q \in I^-(r_k)$. Hence, $I^-(r_k)$ is a neighbourhood of $q$, so $I^-(r_k) \cap H^+(S) \neq \emptyset$, so $r_k \in I^+(H^+(S))$. Therefore, by Lemma~\ref{lem:15}, $r_k \in I^+(S)$, but $r_k \not\in \overline{D^+(S)}$. Thus, if $\lambda_k$ would intersect $\overline{D^+(S)}$, it
would also have to intersect the boundary of that set, i.e., $S\cup H^+(S)$ (by~\eqref{eq: boundary of D+}),
and thereby $S$. But then Lemma~\ref{lem:12}, applied to $x=r_k$ would imply that
$\lambda_k$ intersects $H^+(S)$, a contradiction.

It remains to show that $q\in \bar{S}$. Since $q \in \overline{D^+(S)}$ we have $
I^-(q) \subseteq I^-(\overline{D^+(S)}) \subseteq I^-(S)\cup \overline{D^+(S)}$.
It follows that $p_k \in I^-(q) \setminus \overline{D^+(S)} \subseteq I^-(S)$. Let $\alpha_k $ be a timelike curve from $q$ to $p_k$ contained in $U_k$ and extend it to the past to become past inextendible. As $q \in \mathrm{edge}(H^+(S))\sse \overline{H^+(S)} \subseteq \overline{D^+(S)}$, this curve must, by Proposition~\ref{prop:10}, intersect $S$ in a point $z_k$. Since $p_k \in I^-(S)$ and $S$ is achronal any such $z_k$ must lie between $q$ and $p_k$, hence $z_k\in U_k$. Thus $z_k\to q$,
and therefore $q \in \overline{S}$.
\end{proof}

\begin{Lemma}
\label{lem:16}
Let $S$ be achronal. Then the Cauchy horizon $\hpjp$ of $\jp$ is a closed, achronal topological hypersurface.
\end{Lemma}
\begin{proof}
Clearly $\hpjp$ is closed and achronality follows from Lemma~\ref{lem:15}.
By Lemma~\ref{helem} (and Lemma~\ref{newlem}),
$\text{edge}\left(\hpjp \right) \subseteq \text{edge}(\jp )=\emptyset$ (see Lemma~\ref{lem:3} and~\cite[Prop.~A.18]{hawkingc11}), so the claim follows from~\cite[Prop.~A.18]{hawkingc11}.
\end{proof}

\begin{Lemma}
\label{cor:22}
Let $S$ be closed and achronal. Then
$ H^+(\overline{E^+(S)}) \sse \hpjp $.
\end{Lemma}
\begin{proof}
We roughly follow the proof of~\cite[Lemma 9.3.1]{Krie}. Assume to the contrary that there exists
$p\in H^+(\overline{E^+(S)}) \setminus
\hpjp$. Since $\overline{E^+(S)} \sse \partial J^+(S)$ we have $\overline{D^+(\overline{E^+(S)})} \sse
\overline{\dpjp}$, so $p\in I^-(\dpjp)$. Thus there exists $q$ in $I^+(p)\cap \dpjp$ and
because $p\not\in \hpjp$ and $\hpjp $ is closed, we may additionally assume that $q\notin
\hpjp $. This $q$ is in $I^+(H^+(\overline{E^+(S)}) )$ so by Lemma~\ref{lem:15} $q\notin \overline{D^+(\overline{E^+(S)})}
$. Thus by Proposition~\ref{prop:10} there exists a past inextendible timelike curve $\lambda $ starting in $q$ that
never meets $\overline{E^+(S)}$. However, as any such curve must meet $\partial J^+(S)$ there
exists $z\in \lambda$ with $z\in \partial J^+(S)\setminus E^+(S)$. By Proposition~\ref{prop:4} there exists a past inextendible null curve $\mu \sse
\partial J^+(S)\setminus E^+(S)$ starting in $z$. Finally by Lemma~\ref{prop:11} the concatenation of $\lambda $ and $\mu
$ must enter $I^-(\partial J^+(S))$, contradicting the achronality of $\partial J^+(S)$.
\end{proof}

\subsection{Strong causality}

Finally we are going to collect some results concerning strong causality.
\begin{definition}
\label{def: strong causalty}
\emph{Strong causality} holds at a point
$p \in M$ if for every neighbourhood $U$ of $p$ there exists a neighbourhood
$V$ of $p$ with $V \subseteq U$ such that every causal curve in $M$
that starts and ends in $V$ is entirely contained in $U$.
\end{definition}
As in the smooth case there is the following alternative definition.
\begin{Lemma}
\label{lem:18 (strong causality equiv)}
Strong causality holds at $p$
if and only if for every neighbourhood $U$ of $p$ there exists a
neighbourhood $V$ of $p$ with $V \subseteq U$ such that no causal curve
in $M$ intersects $V$ more than once.
\end{Lemma}
\begin{proof}
See~\cite[Lem.~3.21]{Minguzzi08thecausal}.
\end{proof}

\begin{Lemma}
\label{lem: 19 conj. pts imply strong causality}
If $M$ is chronological and every inextendible
null geodesic is not maximising, then strong causality holds throughout
$M$.
\end{Lemma}
\begin{proof}
The proof is similar to the smooth case, see, e.g.,~\cite[Prop.~12.39]{BEE} or~\cite[Lem.~8.3.7]{Krie}.
Assume to the contrary that strong causality does not hold at some point $p \in M$.
Then there exists a neighbourhood $U$ of $p$ and neighbourhoods $V_{k}$
of $p$ with $\bigcap_{k \in \mathbb{N}}V_{k}=\{p\}$ and future directed
causal curves $\gamma_{k}^{+}$ parametrised with respect to $h$-arclength
that start at $p_{k}=\gamma_{k}^{+}(0) \in V_{k}$ and end at $q_{k}=\gamma_{k}^{+}(b_{k}) \in V_{k}$
but leave $U$. Hence by Theorem~\ref{prop:MinguzziLimit}, there exists a causal limit curve $\gamma^{+}$ starting at $p$.
We may assume that this limit curve is future inextendible: Otherwise $b_k \to b <\infty $ and
$p=\lim_{k\to \infty}\gamma_{k}^{+}(b_{k})=\gamma^{+}(b)$, so
$\gamma^{+}$ is a closed causal curve. But
then Lemma~\ref{pu} and Lemma~\ref{pu:C11} show that two points on $\gamma^{+}$
could be connected by a timelike curve because no inextendible null geodesic is maximising
by assumption, contradicting chronology.

By the same argument, only using the (also future directed) curves $\gamma_k^- \colon [-b_k,0]\to M$ defined by $\gamma_k^-(t):=\gamma_k^+(b_k+t)$, one obtains a past inextendible causal limit curve $\gamma^{-}$ starting at $p$. Together these
two limit curves form an inextendible causal curve $\gamma$.

Since $\gamma$ is inextendible there are points $x=\gamma(t^{-})$ and $y=\gamma(t^{+})$ on $\gamma$
that can be connected by a timelike curve.
We may assume $y \in J^{+}(p)$ and $x \in J^{-}(p)$ by Lemma~\ref{pu} and $\gamma_{k}^{-}(t^{-})\to\gamma(t^{-})$
and $\gamma_{k}^{+}(t^{+})\to\gamma(t^{+})$. Since the relation $\ll$ is open (see~\cite[Sec.\ 1.4]{M} or~\cite[Cor.~3.12]{KSSV}) this implies $\gamma_{k}^{-}(t^{-})\ll\gamma_{k}^{+}(t^{+})$ for $k$ large. Then $\gamma_{k}^{-}(t^{-})=\gamma_{k}^{+}(t^{-}+b_{k})\ll\gamma_{k}^{+}(t^{+})$
and by $b_{k}\to \infty$ we get $t^{-}+b_{k}>t^{+}$ for large enough
$k$, but this yields $\gamma_{k}^+(t^{+})\leq\gamma_{k}^+(t^{-}+b_{k})\ll\gamma_{k}^{+}(t^{+})$,
hence there exists a closed timelike curve through $\gamma_{k}^{+}(t^{+})$, contradicting
chronology of $M$.
\end{proof}

As already remarked in~\cite[Def.~2.6]{S14}, the proof of~\cite[Lem.~14.13]{ON83} remains true even for continuous metrics and so strong causality implies that the spacetime is both \emph{non-totally} and \emph{non-partially imprisoning}, meaning that no future (or past) inextendible causal curve can remain in a compact set or return to it infinitely often. This gives

\begin{Lemma}
\label{lem: 28 the new one}
Let $M$ be strongly causal and let $\gamma$ be an inextendible causal curve in $M$. Then (the image of) $\gamma$
is a closed subset of $M$.
\end{Lemma}

\noindent{\em Acknowledgements.} We are greatly indebted to James Vickers for several discussions that
have importantly contributed to this work. We also thank Clemens S\"amann for valuable input. The work of JG was partially supported by STFC Consolidated Grant ST/L000490/1. MG is the recipient of a DOC Fellowship of the Austrian Academy of Sciences. This work was supported by project P28770 of the Austrian Science Fund FWF. Finally, we gratefully acknowledge the kind hospitality of the Erwin Schrödinger Institute ESI
during the thematic programme ``Geometry and Relativity''.

\end{document}